\documentclass[aps,pra,11pt,onecolumn,nofootinbib,tightenlines,superscriptaddress]{revtex4-2}

\usepackage[T1]{fontenc}
\usepackage{lmodern}
\usepackage{microtype}
\usepackage{amsmath,amssymb,mathtools,amsthm}
\usepackage{braket}
\usepackage{array,booktabs}
\usepackage[table]{xcolor}
\usepackage{placeins}
\usepackage[colorlinks=true,urlcolor=blue,citecolor=blue,linkcolor=blue]{hyperref}

\definecolor{verylightgray}{gray}{0.96}
\definecolor{morelightgray}{gray}{0.975}

\usepackage{tikz}
\usetikzlibrary{positioning, arrows.meta, decorations.pathmorphing, fit}

\newtheorem{theorem}{Theorem}[section]
\newtheorem{lemma}[theorem]{Lemma}
\newtheorem{proposition}[theorem]{Proposition}

\newtheorem{remark}[theorem]{Remark}
\renewcommand{\thetheorem}{\arabic{section}.\arabic{theorem}}
\let\oldremark\remark
\renewcommand{\remark}{\oldremark\upshape}

\renewcommand{\P}{\mathcal{P}}
\newcommand{\norm}[1]{\left\lVert#1\right\rVert}
\DeclareMathOperator{\Tr}{Tr}
\DeclareMathOperator{\supp}{supp}

\begin{document}

\title{Strong Converse Exponent of Quantum State Merging}
\hypersetup{
 pdftitle={Strong Converse Exponent of Quantum State Merging},
 pdfauthor={Mario Berta, Hao-Chung Cheng, Marco Tomamichel},
 pdfkeywords={quantum state merging, strong converse exponent, Renyi entropy, partial smoothing}
}

\begin{abstract}
    We determine the strong converse exponent for the entanglement cost of quantum state merging, showing that it is characterized by the optimized $\alpha$-$z$ conditional R\'enyi entropies with $z=\alpha/2\in[1/2,1]$. This contrasts with the sandwiched conditional R\'enyi entropies that typically govern strong converse exponents in quantum information theory. As a consequence, we derive the strong converse exponent of the partially smoothed conditional min-entropy in purified distance. This exponent is governed by club-sandwiched conditional entropies, whereas global smoothing in purified distance leads to a sandwiched expression.
\end{abstract}

\author{Mario Berta}
\affiliation{Institute for Quantum Information, RWTH Aachen University, Germany}
\affiliation{Department of Computing, Imperial College London, United Kingdom}

\author{Hao-Chung Cheng}
\affiliation{Department of Electrical Engineering, National Taiwan University, Taipei 10617, Taiwan}
\affiliation{Physics/Mathematics Division, National Center for Theoretical Sciences, Taiwan}
\affiliation{Hon Hai (Foxconn) Quantum Computing Center, Taiwan}

\author{Roberto Rubboli}
\affiliation{Department of Mathematical Sciences, University of Copenhagen, Universitetsparken 5, 2100 Denmark}

\author{Marco Tomamichel}
\affiliation{Department of Electrical and Computer Engineering,
National University of Singapore, Singapore 117583, Singapore}
\affiliation{Centre for Quantum Technologies, National University of Singapore, Singapore 117543, Singapore}

\maketitle

\section{Introduction}
\label{sec:introduction}

Quantum state merging is the task of transferring Alice's share $A$ of a tripartite pure state $\psi_{ABR}$ to Bob while preserving the global state, including its coherence with the inaccessible reference $R$.  In the one-way setting, Alice and Bob may use local quantum operations, arbitrary forward classical communication, and preshared entanglement.
A schematic illustration of the one-way quantum state merging task is shown in Figure~\ref{fig:QSM}. 
Introduced by Horodecki, Oppenheim, and Winter~\cite{horodecki2005partial,horodecki2007merging}, state merging gives the conditional entropy a direct operational meaning: the optimal asymptotic net entanglement cost is $H(A|B)_\psi$, and a negative conditional entropy means that the protocol generates rather than consumes entanglement.  State merging has since become a fundamental primitive of quantum Shannon theory and a basic ingredient in distributed quantum compression, coding with quantum side information, assisted entanglement distillation, and a family of decoupling-based quantum protocols~\cite{horodecki2007merging,abeyesinghe2009mother}.
Its one-shot formulations have also played an important role in finite-resource quantum information theory~\cite{berta2009single,anshu2020partially,berta2026tight}.
Since state merging is an intrinsically fully quantum task whose success requires preserving coherence with the inaccessible reference system, fidelity (equivalently, purified distance) is a particularly natural performance criterion, which we adopt throughout this work.
As our results reveal, this natural choice has nontrivial consequences beyond the first-order rate, giving rise to the unconventional R\'enyi structures that govern the corresponding exponents.

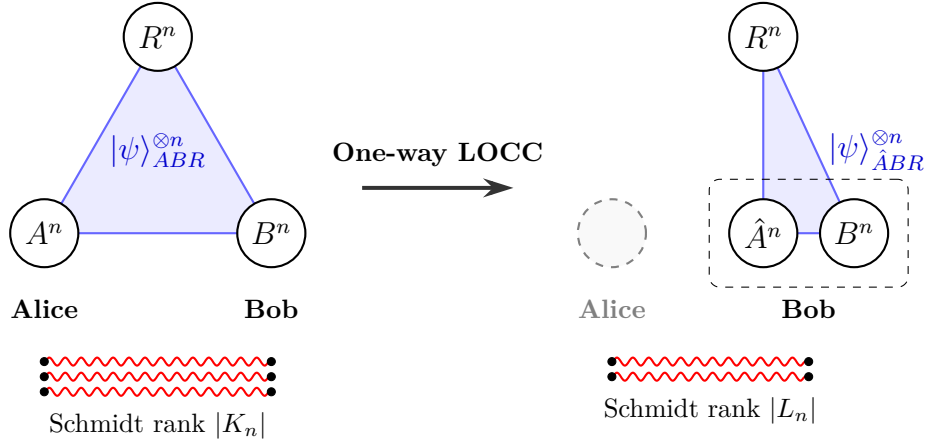
\begin{figure}[htbp]
    \centering
    \begin{tikzpicture}[
        sys/.style={circle, draw=black, thick, fill=white, minimum size=0.9cm, inner sep=1pt, font=\large},
        ent/.style={thick, draw=red, decorate, decoration={snake, amplitude=1.5pt, segment length=6pt}},
        statebg/.style={fill=blue!10, draw=blue!60, thick, rounded corners=6pt, fill opacity=0.8},
    ]

    \coordinate (R_L) at (1.5, 3.8);
    \coordinate (A_L) at (0, 1.2);
    \coordinate (B_L) at (3, 1.2);

    \draw[statebg] (A_L) -- (B_L) -- (R_L) -- cycle;
    \node[color=blue!80!black, font=\large] at (1.5, 2.3) {$|\psi\rangle_{ABR}^{\otimes n}$};

    \node[sys] (node_R_L) at (R_L) {$R^n$};
    \node[sys] (node_A_L) at (A_L) {$A^n$};
    \node[sys] (node_B_L) at (B_L) {$B^n$};

    \node[font=\bfseries] at (0, 0.2) {Alice};
    \node[font=\bfseries] at (3, 0.2) {Bob};

    \draw[ent] (0, -0.5) -- (3, -0.5);
    \draw[ent] (0, -0.7) -- (3, -0.7);
    \draw[ent] (0, -0.9) -- (3, -0.9);
    
    \filldraw (0, -0.5) circle (1.5pt); \filldraw (3, -0.5) circle (1.5pt);
    \filldraw (0, -0.7) circle (1.5pt); \filldraw (3, -0.7) circle (1.5pt);
    \filldraw (0, -0.9) circle (1.5pt); \filldraw (3, -0.9) circle (1.5pt);

    \node[below=0.15cm of {(1.5, -0.9)}] {Schmidt rank $|K_n|$};

    \draw[-{Stealth[length=4mm, width=2.5mm]}, line width=1.5pt, draw=black!80] 
        (4.2, 1.8) -- (6.2, 1.8) 
        node[midway, above=0.15cm, font=\bfseries, text=black] {One-way LOCC};

    \coordinate (A_R) at (7.5, 1.2);       
    
    \coordinate (B_R1) at (9.5, 1.2);      
    \coordinate (B_R2) at (10.7, 1.2);     
    \coordinate (R_R) at (9.5, 3.8);       

    \draw[statebg] (B_R1) -- (B_R2) -- (R_R) -- cycle;
    
    \node[color=blue!80!black, font=\large] at (11.0, 2.3) {$|\psi\rangle_{\hat{A}BR}^{\otimes n}$};

    \node[sys] (node_R_R) at (R_R) {$R^n$};
    \node[sys, dashed, draw=gray, fill=gray!5] (node_A_R) at (A_R) {}; 
    \node[sys] (node_B_R1) at (B_R1) {$\hat{A}^n$};
    \node[sys] (node_B_R2) at (B_R2) {$B^n$};

    \node[draw, dashed, rounded corners, inner sep=0.25cm, fit=(node_B_R1) (node_B_R2)] (bob_box) {};

    \node[font=\bfseries, text=gray] at (7.5, 0.2) {Alice};
    \node[font=\bfseries] at (10.1, 0.2) {Bob};
    
    \draw[ent] (7.5, -0.5) -- (10.1, -0.5);
    \draw[ent] (7.5, -0.7) -- (10.1, -0.7);
    
    \filldraw (7.5, -0.5) circle (1.5pt); \filldraw (10.1, -0.5) circle (1.5pt);
    \filldraw (7.5, -0.7) circle (1.5pt); \filldraw (10.1, -0.7) circle (1.5pt);

    \node[below=0.15cm of {(8.8, -0.7)}] {Schmidt rank $|L_n|$};

    \end{tikzpicture}
    \caption{Illustration of the quantum state merging task protocol.
    Left: Alice and Bob share $n$ copies of a tripartite pure state $|\psi\rangle_{ABR}^{\otimes n}$ alongside preshared entanglement of Schmidt rank $|K_n|$. Right: After applying a one-way local operation and classical communication (LOCC) protocol, Alice's system $A^n$ is merged at Bob's side, consuming or distilling entanglement such that the final entangled state has a Schmidt rank of $|L_n|$.}
    \label{fig:QSM}
\end{figure}

For an independent and identically distributed (i.i.d.) source $\psi_{ABR}^{\otimes n}$, let $r$ denote the net entanglement cost per copy.  Whenever $r>H(A|B)_\psi$, the optimal purified-distance error vanishes exponentially in $n$, as does the infidelity.  Berta, Cheng, and Yao determined the exact error exponent in the low-entanglement-cost portion of this direct regime, expressed in terms of the Petz conditional R\'enyi entropy~\cite{berta2026tight}.  In this paper, we study the complementary \emph{strong converse regime} $r<H(A|B)_\psi$, where the optimal fidelity converges to zero exponentially.  Sharma established an exponential strong converse for state merging~\cite{sharma2014strong}, and Leditzky, Wilde, and Datta obtained related R\'enyi strong-converse bounds for coherent state merging and other protocols~\cite{leditzky2016strong}; however, the exact squared-fidelity decay rate remained open.
Exact strong converse exponents in quantum hypothesis testing, classical-quantum channel coding, entanglement-assisted communication, smoothing, and catalytic decoupling are often governed by sandwiched R\'enyi quantities~\cite{mosonyiogawa2015hypothesis,mosonyi2017strong,li2023strong,li2024operational}.  
In contrast, we show that state merging exhibits a different R\'enyi structure.
Writing $F_n^{\mathrm{merg}}(\psi,r)$ for the optimal squared fidelity under net entanglement cost at most $nr$, we prove
\begin{align}
 \lim_{n\to\infty}-\frac{1}{n}\log F_n^{\mathrm{merg}}(\psi,r)
 &=\max_{\alpha\in[1,2]}\frac{2(\alpha-1)}{\alpha}
 \left[H_{\alpha,\alpha/2}^{\uparrow}(A|B)_\psi-r\right],
 \label{eq:introduction-main-result}
\end{align}
where $H_{\alpha,\alpha/2}^{\uparrow}(A|B)_\psi$ is the $\alpha$-$z$ conditional R\'enyi entropy defined in Eq.~\eqref{eq:alpha-z-conditional}.
The exponent is positive when $r<H(A|B)_\psi$.
For $1<\alpha<2$, the parameter choice $z=\alpha/2$ belongs to neither the Petz family $z=1$ nor the sandwiched family $z=\alpha$.  
Hence, to our knowledge, Eq.~\eqref{eq:introduction-main-result} provides the first operational interpretation of this $\alpha$-$z$ path for exponents in a general quantum information processing task~\cite{audenaert13_alphaz}.\footnote{Verhagen, Tomamichel, and Haapasalo obtained operational interpretations of $\alpha$-$z$ relative entropies with $\alpha<1$ for large-sample and catalytic relative majorization of flat-state pairs and extensions~\cite{verhagen2026conditions}.}

For the converse, we use the one-shot reduction of Anshu, Berta, Jain, and Tomamichel to the \emph{partially smoothed conditional min-entropy}~\cite{anshu2020partially}.
The subsequent noncommutative H\"older inequality and auxiliary-state optimization adapt the club-sandwiched converse method developed by Rubboli and Tomamichel for composable randomness extraction against quantum side information~\cite{rubboli2026strong}, together with standard Schatten-norm H\"older theory~\cite{bhatia1997matrix}.  The club-sandwiched duality of Refs.~\cite{rubboli2024quantum,rubboli2026strong} then converts the resulting expression into the $\alpha$-$z$ conditional entropy in Eq.~\eqref{eq:introduction-main-result}.

For achievability, we combine the recent relative-entropy decoupling theorem of Ref.~\cite{berta2026tight} with the Log-Euclidean change-of-measure and double-blocking method introduced by Mosonyi and Ogawa for classical-quantum channel coding~\cite{mosonyi2017strong} and subsequently developed in related strong-converse-exponent problems~\cite{CHDH-2018,CHDH2-2018,li2023strong,li2024operational,berta2024strong,berta2025strong}.
A controlled Uhlmann recovery argument then converts the resulting decoupling construction into a state-merging protocol.  Finally, adapting the fully quantum double-pinching method of Rubboli and Tomamichel~\cite{rubboli2026strong}, we pass from the Log-Euclidean quantity to the club-sandwiched conditional entropy.

The principal noncommutative difficulty is that the relevant club-sandwiched quantity simultaneously involves $\psi_{AR}$, its marginal $\psi_R$, and an auxiliary state on $R$; a single pinching does not make these three operators commute.  Our fully quantum construction adapts and extends the two-pinching universal-state architecture of Rubboli and Tomamichel~\cite{rubboli2026strong} from classical-quantum randomness extraction to state merging, with the additional requirement that the state-merging marginal be preserved.

The globally and partially smoothed conditional min-entropies provide an instructive comparison.
Consistently with the fidelity criterion adopted for state merging, we use purified distance as the smoothing metric for both notions.
For a bipartite state $\rho_{AR}$, let $q$ denote the target conditional min-entropy rate.  In the strong converse regime $q>H(A|R)_\rho$, global smoothing is governed by the (non-optimized) arrow-down sandwiched conditional entropy $\widetilde H_\alpha^{\downarrow}(A|R)_\rho$~\cite{li2024operational}, whereas Appendix~\ref{app:partial-smoothing-exponent} shows that imposing the marginal constraint changes the strong converse exponent to one governed by the club-sandwiched conditional entropy $\widetilde H_\alpha^{\lambda}(A|R)_\rho$, $\lambda = \frac{1-2\alpha}{1-\alpha}$~\cite{rubboli2026strong}.
The same club-sandwiched quantity determines the strong converse exponent of quantum state merging.
Hence, global and partial smoothing generally have different strong converse exponents.
In the direct regime, however, the one-shot comparison of Ref.~\cite{anshu2020partially}, together with the global-smoothing exponent of Ref.~\cite{li2022tight} and the state-merging exponent of Ref.~\cite{berta2026tight}, shows that global smoothing, partial smoothing, and state merging all have the same error exponent in the low-entanglement-cost regime. 
Table~\ref{tab:exponent-comparison} summarizes these results under the purified-distance criterion.

\FloatBarrier
\begin{table*}[h]
	\caption{Strong converse and error exponents in the i.i.d.\ setting under the purified-distance criterion $P$.  The strong converse exponent is the decay rate of $1-P$ (equivalently, $1-P^2$, which is the squared fidelity), whereas the error exponent is the decay rate of $P$.
	Here $q$ is the target conditional min-entropy rate.  In the state-merging row, $q=-r$ is the net entanglement-distillation rate and $r$ is the net entanglement-cost rate.  
	The error exponent of the globally smoothed conditional min-entropy $H_{\min}^{\varepsilon,P}(A|R)_\rho$ defined in \eqref{eq:globally-smoothed-min} follows directly from Ref.~\cite[Theorem~1]{li2022tight}.
	The error exponent of the partially smoothed conditional min-entropy $H_{\min}^{\varepsilon,P}(A|\dot R)_\rho$ defined in \eqref{eq:partially-smoothed-min} follows by combining Ref.~\cite[Theorem~1]{li2022tight}, Ref.~\cite[Theorem~6]{anshu2020partially}, and Ref.~\cite[Theorem~3]{berta2026tight}.
	The symbol $*$ marks formulas known to be tight only in the low-entanglement-cost regime of Ref.~\cite{berta2026tight}.
	When expressed on the $AR$ system, the three error exponents have the same form.  By contrast, the strong converse exponent for quantum state merging matches the club-sandwiched form~\cite{rubboli2026strong} of the partially smoothed conditional min-entropy rather than the form of the globally smoothed conditional min-entropy.
	}
	\label{tab:exponent-comparison}
	\centering
	\renewcommand{\arraystretch}{1.16}
	\setlength{\tabcolsep}{2pt}
	\scriptsize
	\begin{tabular}{@{}lcc@{}}
		\toprule
		\textbf{Setting} & \textbf{Strong converse exponent} & \textbf{Error exponent} \\
		\midrule
		\begin{tabular}[c]{@{}l@{}}Globally smoothed\\conditional min-entropy\end{tabular}
		&
		\begin{tabular}[c]{@{}c@{}}
			$\displaystyle \max_{\alpha\in[1/2,1]}\frac{1-\alpha}{\alpha}\left[q-\widetilde H_\alpha^{\downarrow}(A|R)_\rho\right]$ \\[2pt]
			\cite{li2024operational}
		\end{tabular}
		&
		\colorbox{morelightgray}{%
			\makebox[5.5cm][c]{%
		\begin{tabular}[c]{@{}c@{}}
			$\displaystyle \frac{1}{2}\max_{\alpha\in[1,\infty]}(\alpha-1)\left[\widetilde H_\alpha^{\downarrow}(A|R)_\rho-q\right]$ \\[2pt]
			\cite{li2022tight}
		\end{tabular}
		}}
		\\
		\arrayrulecolor{lightgray}\midrule\arrayrulecolor{black}
		\begin{tabular}[c]{@{}l@{}}Partially smoothed\\conditional min-entropy\end{tabular}
		&
		\colorbox{verylightgray}{%
			\makebox[5.5cm][c]{%
				\begin{tabular}[c]{@{}c@{}}
					$\displaystyle \max_{\alpha\in[1/2,1]}\frac{1-\alpha}{\alpha}\left[q-\widetilde H_\alpha^{\frac{1-2\alpha}{1-\alpha}}(A|R)_\rho\right]$ \\[2pt]
					(Theorem~\ref{thm:partial-smoothing-exponent})
				\end{tabular}%
		}}
		&
		\colorbox{morelightgray}{%
				\makebox[5.5cm][c]{%
		\begin{tabular}[c]{@{}c@{}}
			$\displaystyle \frac{1}{2}\max_{\alpha\in[1,2]}(\alpha-1)\left[\widetilde H_\alpha^{\downarrow}(A|R)_\rho-q\right]$ \\[2pt]
			\cite{anshu2020partially,li2022tight,berta2026tight}$^*$
		\end{tabular}
		}}
		\\
		\arrayrulecolor{lightgray}\midrule\arrayrulecolor{black}
		\begin{tabular}[c]{@{}l@{}}Quantum state merging\end{tabular}
		&
		\colorbox{verylightgray}{%
			\makebox[5.5cm][c]{%
				\begin{tabular}[c]{@{}c@{}}
					$\displaystyle \max_{\alpha\in[1/2,1]}\frac{1-\alpha}{\alpha}\left[q-\widetilde H_\alpha^{\frac{1-2\alpha}{1-\alpha}}(A|R)_\psi\right]$ \\[10pt]
					$\displaystyle =\max_{\alpha\in[1,2]}\frac{2(\alpha-1)}{\alpha}\left[H_{\alpha,\alpha/2}^{\uparrow}(A|B)_\psi-r\right]$ \\[2pt]
					(Theorem~\ref{thm:main})
				\end{tabular}%
		}}
		&
		\colorbox{morelightgray}{%
				\makebox[5.5cm][c]{%
		\begin{tabular}[c]{@{}c@{}}
			$\displaystyle \frac{1}{2}\max_{\alpha\in[1,2]}(\alpha-1)\left[\widetilde H_\alpha^{\downarrow}(A|R)_\psi-q\right]$ \\[10pt]
			$\displaystyle =\frac{1}{2}\max_{\alpha\in[1/2,1]}\frac{1-\alpha}{\alpha}\left[r-H_{\alpha,1}^{\uparrow}(A|B)_\psi\right]$ \\[2pt]
			\cite{berta2026tight}$^*$
		\end{tabular}
		}}
		\\
		\bottomrule
	\end{tabular}
\end{table*}
\FloatBarrier

The remainder of the paper is organized as follows.  Section~\ref{sec:preliminaries} introduces the entropic quantities, smoothing conventions, and preliminary results used throughout.  Section~\ref{sec:state-merging} defines the state-merging protocol and states the main theorem.  Section~\ref{sec:converse} proves the converse through the partially smoothed conditional min-entropy.  Section~\ref{sec:achievability} establishes achievability using relative-entropy decoupling, a Log-Euclidean change of measure, controlled recovery, and double pinching.
Section~\ref{sec:conclusion} concludes the paper with a comparison between global and partial smoothing.
Appendix~\ref{app:partial-smoothing-exponent} derives the exact strong converse exponent of the partially smoothed conditional min-entropy.

\section{Entropic quantities and preliminaries}
\label{sec:preliminaries}

\subsection{Notation and basic entropies}

All Hilbert spaces are finite-dimensional, and $\mathbb{N}:=\{1,2,\ldots\}$.  All logarithms are base 2, and for a Hermitian operator $X$ we write $\exp X:=2^X$, defined by functional calculus.  We write $|A|$ for the dimension of a system $A$, $\mathcal{S}(A)$ for the set of states on $A$, and $\mathcal{S}_{\leq}(A)$ for the positive semidefinite operators of trace at most one.  Identity operators and identity tensor factors are suppressed when their systems are clear.  The support of a positive semidefinite operator $X$ is denoted by $\supp(X)$.  Powers and logarithms are evaluated on the relevant support; negative powers denote generalized inverses on that support.  For $p\geq1$, $\norm{X}_p:=(\Tr|X|^p)^{1/p}$.  For a scalar $x$, $[x]_+:=\max\{x,0\}$, whereas for a Hermitian operator $X$, $[X]_+:=(|X|+X)/2$ denotes its positive part.

For isomorphic $K$-dimensional systems $K_A$ and $K_B$ with fixed orthonormal bases, we use
\begin{align}
 |\Phi_K\rangle_{K_AK_B}
 &:=\frac{1}{\sqrt K}\sum_{i=1}^K |i\rangle_{K_A}|i\rangle_{K_B},
 &
 \Phi_{K_AK_B}&:=|\Phi_K\rangle\!\langle\Phi_K|.
 \label{eq:maximally-entangled-state}
\end{align}
The von Neumann entropy is
\begin{align}
 H(A)_\rho:=-\Tr[\rho_A\log\rho_A].
 \label{eq:von-Neumann-entropy}
\end{align}

For $\rho,\sigma\in\mathcal{S}_{\leq}(A)$, the generalized squared fidelity and purified distance are
\begin{align}
 F(\rho,\sigma)
 &:=\left(\norm{\sqrt{\rho}\sqrt{\sigma}}_1
 +\sqrt{(1-\Tr\rho)(1-\Tr\sigma)}\right)^2,
 \label{eq:generalized-fidelity}\\
 P(\rho,\sigma)&:=\sqrt{1-F(\rho,\sigma)}.
 \label{eq:purified-distance}
\end{align}
When both arguments are normalized, $F(\rho,\sigma)$ reduces to the squared Uhlmann fidelity. 
For positive semidefinite operators $X$ and $Y$, we use the homogeneous Umegaki relative entropy
\begin{align}
 D(X\|Y):=
 \begin{cases}
 \Tr\left[X(\log X-\log Y)\right],&\supp(X)\subseteq\supp(Y),\\
 +\infty,&\text{otherwise}
 \end{cases}.
 \label{eq:relative-entropy}
\end{align}
 No trace-correction term such as $-\Tr X+\Tr Y$ is included.  All later uses of $D$ for subnormalized operators, including those in Appendix~\ref{app:partial-smoothing-exponent}, follow this homogeneous convention. 
For a normalized bipartite state, the von Neumann conditional entropy is
\begin{align}
 H(A|B)_\rho
 &:=-D\left(\rho_{AB}\middle\|I_A\otimes\rho_B\right)
 =H(AB)_\rho-H(B)_\rho.
 \label{eq:von-Neumann-conditional-entropy}
\end{align}
The max-relative entropy is
\begin{align}
 D_{\max}(X\|Y):=\inf\{\gamma\in\mathbb{R}:X\leq 2^\gamma Y\}.
 \label{eq:max-relative-entropy}
\end{align}

\subsection{R\'enyi divergences and conditional entropies}

We next collect the R\'enyi quantities used in the exponent formulas.  Whenever a displayed parametrization is indeterminate at an endpoint, its value is understood by the stated continuous extension.

For $\alpha,z>0$ with $\alpha\neq1$, the $\alpha$-$z$ R\'enyi divergence is~\cite{audenaert13_alphaz}
\begin{align}
 D_{\alpha,z}(\rho\|\sigma)
 :=\frac{1}{\alpha-1}\log\Tr\left[
 \left(\rho^{\frac{\alpha}{2z}}
 \sigma^{\frac{1-\alpha}{z}}
 \rho^{\frac{\alpha}{2z}}\right)^z\right],
 \label{eq:alpha-z-divergence}
\end{align}
with the standard support convention.  The sandwiched R\'enyi divergence is $\widetilde D_\alpha=D_{\alpha,\alpha}$~\cite{MDS+13,Wilde3}.
We define the $\alpha$-$z$ conditional R\'enyi entropy
\begin{align}
 H_{\alpha,z}^{\uparrow}(A|B)_\rho
 :=\sup_{\sigma_B\in\mathcal{S}(B)}
 -D_{\alpha,z}\!\left(\rho_{AB}\middle\|I_A\otimes\sigma_B\right).
 \label{eq:alpha-z-conditional}
\end{align}
We also use the fixed-marginal (arrow-down) conditional entropy
\begin{align}
 H_{\alpha,z}^{\downarrow}(A|B)_\rho
 &:=-D_{\alpha,z}\!\left(\rho_{AB}\middle\|I_A\otimes\rho_B\right),
 \label{eq:alpha-z-conditional-down}\\
 \widetilde H_\alpha^{\downarrow}(A|B)_\rho
 &:=H_{\alpha,\alpha}^{\downarrow}(A|B)_\rho.
 \label{eq:sandwiched-conditional-down}
\end{align}
At $\alpha=1$, these quantities are defined by continuity and both equal to $H(A|B)_\rho$.

For $0<\alpha<1$ and $\lambda\leq0$, the club-sandwiched conditional entropy of Ref.~\cite{rubboli2024quantum} is
\begin{align}
 \widetilde H_\alpha^\lambda(A|R)_\rho
 :=\min_{\sigma_R\in\mathcal{S}(R)}\frac{1}{1-\alpha}
 \log\Tr\left[
 \left(
 \rho_{AR}^{1/2}
 \left(I_A\otimes
 \rho_R^{\frac{1-\lambda}{2}\frac{1-\alpha}{\alpha}}
 \sigma_R^{\lambda\frac{1-\alpha}{\alpha}}
 \rho_R^{\frac{1-\lambda}{2}\frac{1-\alpha}{\alpha}}
 \right)
 \rho_{AR}^{1/2}
 \right)^\alpha\right].
 \label{eq:club-entropy}
\end{align}
The objective is set to $+\infty$ unless
$\supp(\rho_{AR})\subseteq\supp(I_A\otimes\sigma_R)$.  With this convention,
negative powers are generalized inverses on the relevant support.  We use the boundary of the data-processing region,
\begin{align}
 \widetilde H_\alpha^{\frac{1-2\alpha}{1-\alpha}}(A|R)_\rho
 =\min_{\sigma_R\in\mathcal{S}(R)}\frac{1}{1-\alpha}
 \log\Tr\left[
 \left(
 \rho_{AR}^{1/2}\rho_R^{1/2}
 \sigma_R^{\frac{1-2\alpha}{\alpha}}
 \rho_R^{1/2}\rho_{AR}^{1/2}
 \right)^\alpha\right].
 \label{eq:club-boundary}
\end{align}
The endpoint values at $\alpha=1/2$ and $\alpha=1$ are defined by the appropriate one-sided limits.  In the parameter range
$\alpha\in[1/2,1)$ and $\frac{1-2\alpha}{1-\alpha}\leq\lambda\leq0$, the club-sandwiched conditional entropy is additive and satisfies data processing under channels on the conditioning system~\cite[Theorems~4.1 and~5.1]{rubboli2024quantum}.  Moreover,
\begin{align}
 \lim_{\alpha\nearrow1}
 \widetilde H_\alpha^{\frac{1-2\alpha}{1-\alpha}}(A|R)_\rho
 &=H(A|R)_\rho,
 \label{eq:club-limit}\\
 \widetilde H_\alpha^{\frac{1-2\alpha}{1-\alpha}}(A|R)_\rho
 &\geq H(A|R)_\rho,
 \qquad \alpha\in[1/2,1].
 \label{eq:club-lower-bound}
\end{align}
These properties are proved in Ref.~\cite[Lemmas~17 and~18]{rubboli2026strong}.

For a normalized state $\eta$, a positive semidefinite operator $\sigma$, and $0<\alpha<1$, the Log-Euclidean R\'enyi divergence is~\cite{audenaert13_alphaz,mosonyi2017strong}
\begin{align}
 D_{\alpha,\infty}(\eta\|\sigma)
 :=\frac{1}{\alpha-1}\log\Tr\left[
 Q\exp\left(Q(\alpha\log\eta+(1-\alpha)\log\sigma)Q\right)\right],
 \label{eq:LE-divergence}
\end{align}
where $Q$ projects onto $\supp(\eta)\cap\supp(\sigma)$.  Its Gibbs variational representation, obtained from the variational formula for quantum relative entropy~\cite{petz1988variational} (see also the use of this representation in Refs.~\cite{mosonyi2017strong,rubboli2026strong}), is
\begin{align}
 D_{\alpha,\infty}(\eta\|\sigma)
 =\frac{1}{1-\alpha}\min_{\substack{\tau\in\mathcal{S}\\
 \supp(\tau)\subseteq\supp(\eta)\cap\supp(\sigma)}}
 \left\{\alpha D(\tau\|\eta)+(1-\alpha)D(\tau\|\sigma)\right\}.
 \label{eq:LE-Gibbs}
\end{align}
At $\alpha=1/2$, the Araki--Lieb--Thirring inequality~\cite{araki1990inequality} gives
\begin{align}
 -\log F(\eta,\sigma)=\widetilde D_{1/2}(\eta\|\sigma)
 \leq D_{1/2,\infty}(\eta\|\sigma).
 \label{eq:fidelity-LE}
\end{align}
The formulas above remain valid when $\sigma$ is subnormalized.  

For $0<\alpha<1$ and $\lambda\leq0$, the associated conditional entropy is
\begin{align}
 H_{\alpha,\infty}^{\lambda}(A|R)_\rho
 &:=\min_{\sigma_R\in\mathcal{S}(R)}\frac{1}{1-\alpha}
 \log\Tr\Big[Q_\rho\exp\Big(
 \alpha\log\rho_{AR}
+(1-\alpha)Q_\rho
 \big((1-\lambda)\log\rho_R+\lambda\log\sigma_R\big)Q_\rho
 \Big)\Big],
 \label{eq:LE-conditional}
\end{align}
where $Q_\rho$ projects onto $\supp(\rho_{AR})$ and identities on $A$ are omitted inside the logarithm.
The objective is set to $+\infty$ unless
$\supp(\rho_{AR})\subseteq\supp(I_A\otimes\sigma_R)$, and logarithms are evaluated on the corresponding supports.

\subsection{Smoothing and auxiliary results}

The smoothing quantities and auxiliary one-shot tools below connect the entropic formulas to state merging.  A quantum instrument is a finite family of completely positive, trace-nonincreasing maps whose sum is trace preserving.  We say that it has \emph{one Kraus operator per outcome} when every outcome map has Kraus rank one.

The globally smoothed conditional min-entropy with the conditioning operator fixed to $\rho_R$ is~\cite{anshu2020partially}
\begin{align}
 H_{\min}^{\varepsilon,P}(A|R)_\rho
 :=\max_{\substack{\tau_{AR}\in\mathcal{S}_{\leq}(AR)\\
 P(\tau_{AR},\rho_{AR})\leq\varepsilon}}
 \left[-D_{\max}\!\left(\tau_{AR}\middle\|I_A\otimes\rho_R\right)\right].
 \label{eq:globally-smoothed-min}
\end{align}
The partially smoothed conditional min-entropy that enters the merging converse also imposes the marginal constraint $\tau_R\leq\rho_R$:
\begin{align}
 H_{\min}^{\varepsilon,P}(A|\dot R)_\rho
 :=\max_{\substack{\tau_{AR}\in\mathcal{S}_{\leq}(AR)\\
 P(\tau_{AR},\rho_{AR})\leq\varepsilon,\;\tau_R\leq\rho_R}}
 \left[-D_{\max}\!\left(\tau_{AR}\middle\|I_A\otimes\rho_R\right)\right].
 \label{eq:partially-smoothed-min}
\end{align}

\begin{proposition}[Partially smoothed one-way merging converse {\cite{anshu2020partially}}]
\label{prop:partially-smoothed-converse}
Let $\psi_{ABR}$ be pure.  If a one-way merging protocol has purified-distance error at most $\varepsilon\in(0,1]$, initial and final Schmidt ranks $K$ and $L$, and net entanglement cost $c=\log K-\log L$, then
\begin{align}
 c\geq-H_{\min}^{\varepsilon,P}(A|\dot R)_\psi.
 \label{eq:partially-smoothed-converse}
\end{align}
\end{proposition}

Proposition~\ref{prop:partially-smoothed-converse} is the converse part of Theorem~6 of Ref.~\cite{anshu2020partially}, translated to the squared-fidelity convention in \eqref{eq:generalized-fidelity}.

\begin{proposition}[Fixed-output relative-entropy decoupling {\cite{berta2026tight}}]
\label{prop:fixed-output-decoupling}
Let $\tau_{AR}\in\mathcal{S}(AR)$ and $q<H(A|R)_\tau$.  There exist finite-dimensional registers $K_{A,n}$ and $L_{A,n}$, a finite classical register $J_n$, and a one-Kraus-per-outcome instrument
$\mathcal{A}_n:A^nK_{A,n}\to J_nL_{A,n}$, such that
\begin{align}
 \log|L_{A,n}|-\log|K_{A,n}|&\geq nq,
 \label{eq:decoupling-gain}\\
 \frac{1}{n}\left(\log|L_{A,n}|-\log|K_{A,n}|\right)&\longrightarrow q,
 \label{eq:decoupling-gain-limit}
\end{align}
and, for a constant $c=c(\tau,q)>0$ independent of $n$ and all sufficiently large $n$,
\begin{align}
 D\!\left(
 (\mathcal{A}_n\otimes\mathrm{id}_{R^n})
 (\tau_{AR}^{\otimes n}\otimes\pi_{K_{A,n}})
 \middle\|
 \pi_{J_n}\otimes\pi_{L_{A,n}}\otimes\tau_R^{\otimes n}
 \right)
 \leq2^{-cn}.
 \label{eq:fixed-output-decoupling}
\end{align}
The statement permits both isometric padding of $A^nK_{A,n}$ and an isometric embedding of the physical output into $L_{A,n}$.
\end{proposition}

This is the equal-rank partial-isometry specialization of Lemma~9 and the construction in the proof of Theorem~10 of Ref.~\cite{berta2026tight}.  We use Proposition~\ref{prop:fixed-output-decoupling} as an existing decoupling theorem and do not reproduce its proof.

\begin{proposition}[Log-Euclidean conditional variational formula {\cite{rubboli2026strong}}]
\label{prop:LE-conditional-variational}
For every $\rho_{AR}\in\mathcal{S}(AR)$, $0<\alpha<1$, and $\lambda\leq0$,
\begin{align}
 (\alpha-1)H_{\alpha,\infty}^{\lambda}(A|R)_\rho
 &=\min_{\tau_{AR}\in\mathcal{S}(AR)}\Big\{
 (1-\lambda)(1-\alpha)D(\tau_R\|\rho_R)
+\alpha D(\tau_{AR}\|\rho_{AR})
 +(\alpha-1)H(A|R)_\tau\Big\}.
 \label{eq:LE-conditional-variational}
\end{align}
\end{proposition}

Proposition~\ref{prop:LE-conditional-variational} is Lemma~7 of Ref.~\cite{rubboli2026strong}; its statement is valid for every bipartite quantum state.

\begin{proposition}[Club-sandwiched duality {\cite{rubboli2024quantum,rubboli2026strong}}]
\label{prop:club-duality}
Let $\psi_{ABR}$ be pure.  For $\alpha\in[1/2,1]$, with endpoint values defined by continuity,
\begin{align}
 \widetilde H_\alpha^{\frac{1-2\alpha}{1-\alpha}}(A|R)_\psi
 =-H_{\frac{2\alpha}{3\alpha-1},\frac{\alpha}{3\alpha-1}}^{\uparrow}(A|B)_\psi.
 \label{eq:club-duality}
\end{align}
\end{proposition}

This is the duality relation used in Lemma~16 of Ref.~\cite{rubboli2026strong}; it is a specialization of the general duality theory in Ref.~\cite{rubboli2024quantum}.

For a positive semidefinite operator $X=\sum_i x_iP_i$ with distinct eigenvalues $x_i$, let
\begin{align}
 \P_X(Y):=\sum_iP_iYP_i
 \label{eq:pinching-map}
\end{align}
be the pinching with respect to $X$.  We write $|\operatorname{spec}(X)|$ for the number of distinct eigenvalues of $X$.

\begin{proposition}[Universal state and symmetric optimizer {\cite{hayashitomamichel15c,rubboli2026strong}}]
\label{prop:universal-state}
Let $R$ have dimension $d_R$.  For every $m\in\mathbb{N}$, there exists a permutation-invariant universal state $\omega_{R^m}$ such that every permutation-invariant state $\sigma_{R^m}$ satisfies
\begin{align}
 \sigma_{R^m}\leq(m+1)^{d_R^2-1}\omega_{R^m}.
 \label{eq:universal-domination}
\end{align}
Moreover, $\omega_{R^m}$ commutes with every permutation-invariant state, and
\begin{align}
 \max\left\{
 |\operatorname{spec}(\rho_R^{\otimes m})|,
 |\operatorname{spec}(\omega_{R^m})|
 \right\}\leq(m+1)^{d_R-1}.
 \label{eq:spectral-cardinality}
\end{align}
If a bipartite state on $A^mR^m$ is invariant under simultaneous permutations of the $m$ copies, then for every $0<\alpha<1$ and $\lambda\leq0$, an optimizer in \eqref{eq:LE-conditional} can be chosen permutation-invariant on $R^m$. 
\end{proposition}

The universal-state properties are collected in Lemma~26 of Ref.~\cite{rubboli2026strong}, and the optimizer symmetry is Lemma~20 in the same reference.

\begin{proposition}[Fidelity loss under pinching {\cite{tomamichel16_book,rubboli2026strong}}]
\label{prop:fidelity-pinching}
Let $\P(Y)=\sum_{i=1}^vP_iYP_i$ be a pinching with $v$ spectral projections.  If $\eta\in\mathcal{S}$, $\sigma\in\mathcal{S}_{\leq}$, and $\P(\sigma)=\sigma$, then
\begin{align}
 -\log F(\eta,\sigma)
 \leq-\log F(\P(\eta),\sigma)+\log v.
 \label{eq:fidelity-pinching}
\end{align}
\end{proposition}

This is the $\alpha=1/2$ pinching estimate used as Lemma~27 of Ref.~\cite{rubboli2026strong}; see also Lemma~4.11 of Ref.~\cite{tomamichel16_book}.  The same proof applies when the second argument is subnormalized, because the first argument is normalized.

\section{Quantum state merging and the strong converse exponent}
\label{sec:state-merging}

Let $\psi_{ABR}$ be pure, where Alice holds $A$, Bob holds $B$, and $R$ is an inaccessible reference system.  A blocklength-$n$ \emph{one-way quantum state-merging protocol} consists of an initially shared maximally entangled state $\Phi_{K_{A,n}K_{B,n}}$ with $|K_{A,n}|=|K_{B,n}|=K_n$, an instrument $\{\mathcal{E}_{x,n}:A^nK_{A,n}\to L_{A,n}\}_x$ applied by Alice, transmission of the outcome $x$ to Bob, and a channel $\mathcal{D}_{x,n}:B^nK_{B,n}\to\widehat A^nB^nL_{B,n}$ applied by Bob.  Here, $\widehat A$ is an isomorphic copy of $A$ held by Bob.  The output state is
\begin{align}
 \omega^{\mathcal{M}_n}_{L_{A,n}L_{B,n}\widehat A^nB^nR^n}
 :=\sum_x
 (\mathcal{E}_{x,n}\otimes\mathcal{D}_{x,n}\otimes\mathrm{id}_{R^n})
 \left(\psi_{ABR}^{\otimes n}\otimes\Phi_{K_{A,n}K_{B,n}}\right).
 \label{eq:merging-output}
\end{align}
The ideal output is
\begin{align}
 \Phi_{L_{A,n}L_{B,n}}\otimes\psi_{\widehat A BR}^{\otimes n},
 \label{eq:merging-target}
\end{align}
where $|L_{A,n}|=|L_{B,n}|=L_n$ and $\Phi_{L_{A,n}L_{B,n}}$ has Schmidt rank $L_n$.  This is the usual state-merging task introduced and developed in Refs.~\cite{horodecki2005partial,horodecki2007merging}, with arbitrary forward classical communication.

The net entanglement cost of $\mathcal{M}_n$ is
\begin{align}
 c(\mathcal{M}_n):=\log K_n-\log L_n.
 \label{eq:net-entanglement-cost}
\end{align}
Only this difference is constrained; no separate upper bound is imposed on $K_n$ or $L_n$.  For $r\in\mathbb{R}$, define the optimal squared fidelity
\begin{align}
 F_n^{\mathrm{merg}}(\psi,r)
 :=\sup_{\substack{\mathcal{M}_n:\;c(\mathcal{M}_n)\leq nr}}
 F\!\left(
 \omega^{\mathcal{M}_n},
 \Phi_{L_{A,n}L_{B,n}}\otimes\psi_{\widehat A BR}^{\otimes n}
 \right),
 \label{eq:optimal-merging-fidelity}
\end{align}
where the supremum ranges over all finite choices of $K_n$ and $L_n$ and all protocols of the form above.

\begin{theorem}[Exact strong converse exponent]
\label{thm:main}
For every finite-dimensional pure state $\psi_{ABR}$ and every $r\in\mathbb{R}$, the limit
\begin{align}
 E_{\mathrm{sc}}^{\mathrm{merg}}(\psi,r)
 :=\lim_{n\to\infty}-\frac{1}{n}\log F_n^{\mathrm{merg}}(\psi,r)
 \label{eq:strong-converse-exponent-definition}
\end{align}
exists and satisfies
\begin{align}
 \boxed{
  E_{\mathrm{sc}}^{\mathrm{merg}}(\psi,r)
 =\max_{\alpha\in[1,2]}\frac{2(\alpha-1)}{\alpha}
 \left[H_{\alpha,\alpha/2}^{\uparrow}(A|B)_\psi-r\right]
 }
 \label{eq:main-conditional-entropy}
\end{align}
Moreover,
\begin{align}
 E_{\mathrm{sc}}^{\mathrm{merg}}(\psi,r)>0
 \quad\Longleftrightarrow\quad
 r<H(A|B)_\psi.
 \label{eq:positive-region}
\end{align}
\end{theorem}

\noindent The converse direction ($\geq$) in \eqref{eq:main-conditional-entropy} is proved in Section~\ref{sec:converse}, and the achievability direction ($\leq$) is proved in Section~\ref{sec:achievability}.

Writing $q:=-r$ for the net entanglement-distillation rate, the club-sandwiched conditional entropy duality in Proposition~\ref{prop:club-duality} gives the equivalent representation
\begin{align}
 E_{\mathrm{sc}}^{\mathrm{merg}}(\psi,r)
 =\max_{\alpha\in[1/2,1]}\frac{1-\alpha}{\alpha}
 \left[q-
 \widetilde H_\alpha^{\frac{1-2\alpha}{1-\alpha}}(A|R)_\psi\right].
 \label{eq:main-club-form}
\end{align}
The positivity of the strong converse exponent is then
\begin{align}
 E_{\mathrm{sc}}^{\mathrm{merg}}(\psi,r)>0
 \quad\Longleftrightarrow\quad
 q> \lim_{\alpha\nearrow1}  \widetilde H_\alpha^{\frac{1-2\alpha}{1-\alpha}}(A|R)_\psi =  H(A|R)_\psi,
 \label{eq:positive-region-club}
\end{align}
where purity gives $H(A|B)_\psi=-H(A|R)_\psi$.

\section{Converse}
\label{sec:converse}

For $\rho_{AR}\in\mathcal{S}(AR)$ and $q\in\mathbb{R}$, define the fixed-marginal smoothing fidelity
\begin{align}
 \Gamma_{A|\dot R}(\rho,q)
 :=\max_{\substack{\tau_{AR}\in\mathcal{S}_{\leq}(AR)\\
 \tau_R\leq\rho_R,\;\tau_{AR}\leq2^{-q}I_A\otimes\rho_R}}
 F(\rho_{AR},\tau_{AR}).
 \label{eq:Gamma-definition}
\end{align}
The maximum exists by compactness.  By \eqref{eq:partially-smoothed-min} and $P(\rho,\tau)^2=1-F(\rho,\tau)$, for every $\varepsilon\in[0,1]$,
\begin{align}
 H_{\min}^{\varepsilon,P}(A|\dot R)_\rho\geq q
 \quad\Longleftrightarrow\quad
 \Gamma_{A|\dot R}(\rho,q)\geq1-\varepsilon^2.
 \label{eq:Hmin-Gamma-equivalence}
\end{align}

\begin{lemma}[One-shot strong converse for partially smoothed conditional min-entropy]
\label{lem:partial-smoothing-one-shot}
For every $\rho_{AR}\in\mathcal{S}(AR)$, $q\in\mathbb{R}$, and $\alpha\in[1/2,1)$,
\begin{align}
 \Gamma_{A|\dot R}(\rho,q)
 \leq
 2^{-\frac{1-\alpha}{\alpha}
 \left[q-\widetilde H_\alpha^{\frac{1-2\alpha}{1-\alpha}}(A|R)_\rho\right]}.
 \label{eq:Gamma-one-shot-converse}
\end{align}
Equivalently, if $H_{\min}^{\varepsilon,P}(A|\dot R)_\rho\geq q$, then
\begin{align}
 1-\varepsilon^2
 \leq
 2^{-\frac{1-\alpha}{\alpha}
 \left[q-\widetilde H_\alpha^{\frac{1-2\alpha}{1-\alpha}}(A|R)_\rho\right]}.
 \label{eq:Hmin-one-shot-converse}
\end{align}
\end{lemma}

\begin{proof}
Restrict $R$ to $\supp(\rho_R)$ and fix an operator $\tau_{AR}$ feasible in \eqref{eq:Gamma-definition}.  Fix $\alpha\in(1/2,1)$ and a faithful $\sigma_R$.  The following noncommutative H\"older inequality is the fully quantum analogue of the factorization used for the club-sandwiched converse in Ref.~\cite{rubboli2026strong}; we use the Schatten-norm H\"older inequality in the form recorded, for example, in Ref.~\cite{bhatia1997matrix}:
\begin{align}
 \sqrt{F(\rho_{AR},\tau_{AR})}
 &\leq
 \norm{\rho_{AR}^{1/2}
 \left(I_A\otimes\rho_R^{1/2}
 \sigma_R^{\frac{1-2\alpha}{2\alpha}}\right)}_{2\alpha}
 \times
 \norm{\left(I_A\otimes
 \sigma_R^{\frac{2\alpha-1}{2\alpha}}\rho_R^{-1/2}\right)
 \tau_{AR}^{1/2}}_{\frac{2\alpha}{2\alpha-1}}.
 \label{eq:Gamma-holder-factorization}
\end{align}
The square of the first factor equals
\begin{align}
 \left\{\Tr\left[
 \left(
 \rho_{AR}^{1/2}\rho_R^{1/2}
 \sigma_R^{\frac{1-2\alpha}{\alpha}}
 \rho_R^{1/2}\rho_{AR}^{1/2}
 \right)^\alpha\right]\right\}^{1/\alpha}.
 \label{eq:Gamma-first-holder-factor}
\end{align}
For the second factor, set
\begin{align}
 C_{AR}:=(I_A\otimes\rho_R^{-1/2})
 \tau_{AR}(I_A\otimes\rho_R^{-1/2}).
 \label{eq:Gamma-C-definition}
\end{align}
The two constraints in \eqref{eq:Gamma-definition} imply
\begin{align}
 0\leq C_{AR}\leq2^{-q}I_{AR},
 \qquad
 \Tr[(I_A\otimes\sigma_R)C_{AR}]\leq1.
 \label{eq:Gamma-C-bounds}
\end{align}
The Araki--Lieb--Thirring inequality~\cite{araki1990inequality} therefore yields
\begin{align}
\norm{\left(I_A\otimes
 \sigma_R^{\frac{2\alpha-1}{2\alpha}}\rho_R^{-1/2}\right)
 \tau_{AR}^{1/2}}_{\frac{2\alpha}{2\alpha-1}}^{\frac{2\alpha}{2\alpha-1}}
&\leq
 \Tr\left[(I_A\otimes\sigma_R^{1/2})
 C_{AR}^{\frac{\alpha}{2\alpha-1}}
 (I_A\otimes\sigma_R^{1/2})\right]
 \nonumber\\
 &\leq
 2^{-q\frac{1-\alpha}{2\alpha-1}}
 \Tr[(I_A\otimes\sigma_R)C_{AR}]
 \leq2^{-q\frac{1-\alpha}{2\alpha-1}}.
 \label{eq:Gamma-second-holder-factor}
\end{align}
Thus the square of the second factor in \eqref{eq:Gamma-holder-factorization} is at most $2^{-q(1-\alpha)/\alpha}$.  Combining this with \eqref{eq:Gamma-first-holder-factor}, minimizing over $\sigma_R$, and using \eqref{eq:club-boundary} proves \eqref{eq:Gamma-one-shot-converse} for $\alpha\in(1/2,1)$.  The bound for singular $\sigma_R$ follows by full-rank regularization, and the endpoint $\alpha=1/2$ follows by one-sided continuity.  Equation \eqref{eq:Hmin-one-shot-converse} is then immediate from \eqref{eq:Hmin-Gamma-equivalence}.
\end{proof}

\begin{lemma}[One-shot operational reduction]
\label{lem:one-shot-operational-reduction}
Let a state-merging protocol for $\psi_{ABR}$ have initial and final Schmidt ranks $K$ and $L$, coherent entanglement gain
\begin{align}
 q:=\log L-\log K,
 \label{eq:coherent-gain}
\end{align}
and squared fidelity $f$.  Then
\begin{align}
 f\leq\Gamma_{A|\dot R}(\psi_{AR},q).
 \label{eq:operational-Gamma-reduction}
\end{align}
\end{lemma}

\begin{proof}
The claim is trivial when $f=0$.  If $0<f<1$, the purified-distance error is $\varepsilon=\sqrt{1-f}$.  Proposition~\ref{prop:partially-smoothed-converse} gives
\begin{align}
 H_{\min}^{\varepsilon,P}(A|\dot R)_\psi\geq q,
 \label{eq:operational-Hmin-lower-bound}
\end{align}
and \eqref{eq:Hmin-Gamma-equivalence} yields \eqref{eq:operational-Gamma-reduction}.  If $f=1$, the same conclusion follows by applying the argument with any $\varepsilon>0$ and letting $\varepsilon\downarrow0$.
\end{proof}

\begin{theorem}[Finite-block converse and exponent lower bound]
\label{thm:converse}
For every pure state $\psi_{ABR}$, every $r\in\mathbb{R}$, every $n\in\mathbb{N}$, and every $\alpha\in[1,2]$,
\begin{align}
 F_n^{\mathrm{merg}}(\psi,r)
 \leq 2^{-n \frac{2(\alpha-1)}{\alpha}
 \left[H_{\alpha,\alpha/2}^{\uparrow}(A|B)_\psi-r\right]}.
 \label{eq:finite-block-converse}
\end{align}
Consequently,
\begin{align}
 \liminf_{n\to\infty}-\frac{1}{n}\log F_n^{\mathrm{merg}}(\psi,r)
 \geq
 \max_{\alpha\in[1,2]}\frac{2(\alpha-1)}{\alpha}
 \left[H_{\alpha,\alpha/2}^{\uparrow}(A|B)_\psi-r\right].
 \label{eq:converse-exponent}
\end{align}
\end{theorem}

\begin{proof}
Fix $n\in\mathbb{N}$, $\alpha\in[1/2,1)$, and an admissible protocol.  Its coherent gain satisfies
\begin{align}
 \log L_n-\log K_n\geq-nr.
 \label{eq:gain-rate-constraint}
\end{align}
Combining Lemmas~\ref{lem:one-shot-operational-reduction} and~\ref{lem:partial-smoothing-one-shot} with additivity of the club-sandwiched conditional entropy gives
\begin{align}
 F\!\left(
 \omega^{\mathcal{M}_n},
 \Phi_{L_{A,n}L_{B,n}}\otimes\psi_{\widehat A BR}^{\otimes n}
 \right)
 \leq
 2^{-n\frac{1-\alpha}{\alpha}
 \left[-r-
 \widetilde H_\alpha^{\frac{1-2\alpha}{1-\alpha}}(A|R)_\psi\right]}.
 \label{eq:uniform-converse}
\end{align}
The right-hand side is independent of the protocol, so the same bound holds for $F_n^{\mathrm{merg}}(\psi,r)$.  Set
\begin{align}
 \beta:=\frac{2\alpha}{3\alpha-1}.
 \label{eq:parameter-conversion}
\end{align}
As $\alpha$ ranges from $1/2$ to $1$, $\beta$ ranges from $2$ to $1$, and $(1-\alpha)/\alpha=2(\beta-1)/\beta$.  Proposition~\ref{prop:club-duality} therefore converts the protocol-independent one-shot bound into \eqref{eq:finite-block-converse} for $\beta\in(1,2]$.  At $\beta=1$, Eq.~\eqref{eq:finite-block-converse} reduces to the trivial estimate $F_n^{\mathrm{merg}}(\psi,r)\leq1$ and is included by continuity.  Taking negative logarithms, dividing by $n$, taking the liminf, and then optimizing over $\beta\in[1,2]$ yields \eqref{eq:converse-exponent}.
Since $\beta\in[1,2]$ is only a dummy variable, we write $\alpha\in[1,2]$ by convention
\end{proof}

\section{Achievability}
\label{sec:achievability}

The proof reduces the problem to fixed-output decoupling.  Relative-entropy decoupling enters through Proposition~\ref{prop:fixed-output-decoupling}.  We first obtain a Log-Euclidean exponent by a change of measure, then convert fixed-output fidelity into merging fidelity by a controlled Uhlmann argument, and finally use double pinching to replace the Log-Euclidean entropy by the club-sandwiched entropy.

\subsection{Log-Euclidean fixed-output exponent}

\begin{lemma}[Variational identity]
\label{lem:LE-exponent-variational}
For every $\rho_{AR}\in\mathcal{S}(AR)$ and every $q\in\mathbb{R}$,
\begin{align}
 &\max_{\alpha\in[1/2,1]}
 \frac{1-\alpha}{\alpha}
 \left[q-H_{\alpha,\infty}^{\frac{1-2\alpha}{1-\alpha}}(A|R)_\rho\right]
 \nonumber\\
 &\qquad=
 \min_{\tau_{AR}\in\mathcal{S}(AR)}
 \left\{
 D(\tau_{AR}\|\rho_{AR})+D(\tau_R\|\rho_R)
 +[q-H(A|R)_\tau]_+
 \right\}.
 \label{eq:LE-exponent-variational}
\end{align}
\end{lemma}

\begin{proof}
Apply Proposition~\ref{prop:LE-conditional-variational} with
$\lambda=(1-2\alpha)/(1-\alpha)$.  Since
\begin{align}
 (1-\lambda)(1-\alpha)=\alpha,
 \label{eq:LE-path-identity}
\end{align}
setting $t=(1-\alpha)/\alpha$ gives
\begin{align}
 \max_{0\leq t\leq1}\min_{\tau_{AR}\in\mathcal{S}(\supp\rho_{AR})}
 \Bigl\{
 D(\tau_{AR}\|\rho_{AR})+D(\tau_R\|\rho_R)
 +t\bigl(q-H(A|R)_\tau\bigr)
 \Bigr\}.
 \label{eq:LE-minimax}
\end{align}
Here $\mathcal{S}(\supp\rho_{AR})$ denotes the states supported on $\supp(\rho_{AR})$; states outside this set have infinite objective.  The endpoint $\alpha=1$ is represented by $t=0$ through continuous extension.  The state space is compact and convex.  For fixed $t$, the objective is convex and continuous in $\tau_{AR}$ on this support-restricted domain, and for fixed $\tau_{AR}$ it is affine and continuous in $t$.  Sion's minimax theorem~\cite{sion1958general} therefore yields
\begin{align}
 \min_{\tau_{AR}\in\mathcal{S}(\supp\rho_{AR})}
 \Bigl\{
 D(\tau_{AR}\|\rho_{AR})+D(\tau_R\|\rho_R)
 +\max_{0\leq t\leq1}t\bigl(q-H(A|R)_\tau\bigr)
 \Bigr\}.
 \label{eq:LE-minimax-exchanged}
\end{align}
Finally,
\begin{align}
 \max_{0\leq t\leq1}t\bigl(q-H(A|R)_\tau\bigr)
 =[q-H(A|R)_\tau]_+,
 \label{eq:positive-part-identity}
\end{align}
which proves the claim.
\end{proof}

\begin{proposition}[Log-Euclidean fixed-output achievability]
\label{prop:LE-fixed-output}
For every $\rho_{AR}\in\mathcal{S}(AR)$ and $q\in\mathbb{R}$, there exist one-Kraus-per-outcome instruments $\mathcal{A}_n:A^nK_{A,n}\to J_nL_{A,n}$ such that
\begin{align}
 \log|L_{A,n}|-\log|K_{A,n}|&\geq nq,
 \label{eq:LE-achievable-gain}\\
 \frac{1}{n}\left(\log|L_{A,n}|-\log|K_{A,n}|\right)&\longrightarrow q,
 \label{eq:LE-achievable-gain-limit}
\end{align}
and, with
\begin{align}
 \omega_n^\rho
 &:=(\mathcal{A}_n\otimes\mathrm{id}_{R^n})
 (\rho_{AR}^{\otimes n}\otimes\pi_{K_{A,n}}),
 \label{eq:omega-rho}\\
 \zeta_n^\rho
 &:=\pi_{J_n}\otimes\pi_{L_{A,n}}\otimes\rho_R^{\otimes n},
 \label{eq:zeta-rho}
\end{align}
we have
\begin{align}
 \limsup_{n\to\infty}-\frac{1}{n}\log F(\omega_n^\rho,\zeta_n^\rho)
 \leq
 \max_{\alpha\in[1/2,1]}
 \frac{1-\alpha}{\alpha}
 \left[q-H_{\alpha,\infty}^{\frac{1-2\alpha}{1-\alpha}}(A|R)_\rho\right].
 \label{eq:LE-fixed-output-achievability}
\end{align}
\end{proposition}

\begin{proof}
	Choose a minimizer $\tau_{AR}$ in \eqref{eq:LE-exponent-variational}, which exists by compactness and lower semicontinuity.
	
	Suppose first that $q<H(A|R)_\tau$.  Apply Proposition~\ref{prop:fixed-output-decoupling} to $\tau_{AR}$ at coherent gain $q$, and then apply the resulting instrument to $\rho_{AR}^{\otimes n}\otimes\pi_{K_{A,n}}$.  Denote
	\begin{align}
		\omega_n^\tau
		&:=(\mathcal{A}_n\otimes\mathrm{id}_{R^n})
		(\tau_{AR}^{\otimes n}\otimes\pi_{K_{A,n}}),
		&
		\zeta_n^\tau
		&:=\pi_{J_n}\otimes\pi_{L_{A,n}}\otimes\tau_R^{\otimes n}.
		\label{eq:tau-output-and-target}
	\end{align}
	Using \eqref{eq:fidelity-LE} and evaluating the variational formula \eqref{eq:LE-Gibbs} at $\omega_n^\tau$, we obtain
	\begin{align}
		-\log F(\omega_n^\rho,\zeta_n^\rho)
		&\leq D_{1/2,\infty}(\omega_n^\rho\|\zeta_n^\rho)
		\nonumber\\
		&\leq D(\omega_n^\tau\|\omega_n^\rho)
		+D(\omega_n^\tau\|\zeta_n^\rho).
		\label{eq:change-of-measure}
	\end{align}
	Data processing gives
	\begin{align}
		D(\omega_n^\tau\|\omega_n^\rho)
		\leq nD(\tau_{AR}\|\rho_{AR}).
		\label{eq:first-change-of-measure-term}
	\end{align}
	Moreover, $(\omega_n^\tau)_{R^n}=\tau_R^{\otimes n}$, and therefore
	\begin{align}
		D(\omega_n^\tau\|\zeta_n^\rho)
		=D(\omega_n^\tau\|\zeta_n^\tau)
		+nD(\tau_R\|\rho_R).
		\label{eq:fixed-target-chain-rule}
	\end{align}
	The first term on the right of \eqref{eq:fixed-target-chain-rule} converges to zero exponentially by Proposition~\ref{prop:fixed-output-decoupling}.  Hence
	\begin{align}
		\limsup_{n\to\infty}-\frac{1}{n}\log F(\omega_n^\rho,\zeta_n^\rho)
		\leq D(\tau_{AR}\|\rho_{AR})+D(\tau_R\|\rho_R).
		\label{eq:LE-low-gain}
	\end{align}
	
	Suppose now that $q\geq H(A|R)_\tau$.  Choose numbers $\delta_j\downarrow0$.  For each $j$, apply the preceding argument at gain $H(A|R)_\tau-\delta_j$.  By Proposition~\ref{prop:fixed-output-decoupling}, there is an integer $N_j$ such that, for every $n\geq N_j$, the corresponding decoupling estimate holds and the physical coherent gain $g_{n,j}:=n^{-1}(\log L'_{n,j}-\log K_{n,j})$ satisfies
	\begin{align}
		\left|g_{n,j}-\bigl(H(A|R)_\tau-\delta_j\bigr)\right|\leq\delta_j.
		\label{eq:diagonal-gain-control}
	\end{align}
	Take the $N_j$ strictly increasing.  For $n\geq N_1$, set $j(n):=\max\{j:N_j\leq n\}$ and $\delta_n:=\delta_{j(n)}$; for the finitely many $n<N_1$, choose any admissible instrument and set $\delta_n:=\delta_1$.  Then $\delta_n\to0$, and for every $n\geq N_1$ the instrument from the $j(n)$th family satisfies both the decoupling estimate and \eqref{eq:diagonal-gain-control}.  Using the harmless common input--output padding allowed in Proposition~\ref{prop:fixed-output-decoupling}, we may also arrange that integer rounding contributes $o(n)$ to all logarithmic dimensions.
	
	Keep the same initial entanglement rank and isometrically embed the physical output of dimension $L_n':=L'_{n,j(n)}$ into a declared system of dimension
	\begin{align}
		L_n:=\left\lceil |K_{A,n}|2^{nq}\right\rceil.
		\label{eq:declared-output-dimension}
	\end{align}
	For all sufficiently large $n$, $L_n\geq L_n'$.  Let $\omega_{n,\mathrm{phys}}^\rho$ denote the source output on $J_nL'_{A,n}R^n$, and set $\zeta_{n,\mathrm{phys}}^\rho:=\pi_{J_n}\otimes\pi_{L_n'}\otimes\rho_R^{\otimes n}$.  Let $\iota_n:L_{A,n}'\hookrightarrow L_{A,n}$ be an isometric embedding and compose the physical instrument with $\mathrm{id}_{J_n}\otimes\iota_n$.  The resulting declared output and target are
	\begin{align}
		\omega_n^\rho
		&:= (\mathrm{id}_{J_n}\otimes\iota_n\otimes\mathrm{id}_{R^n})
		\omega_{n,\mathrm{phys}}^\rho
		(\mathrm{id}_{J_n}\otimes\iota_n^\dagger\otimes\mathrm{id}_{R^n}),
		&
		\zeta_n^\rho
		&:=\pi_{J_n}\otimes\pi_{L_n}\otimes\rho_R^{\otimes n}.
		\label{eq:declared-embedded-output}
	\end{align}
	Since $\omega_n^\rho$ remains supported on the embedded $L_n'$-dimensional subspace, squared fidelity satisfies
	\begin{align}
		F(\omega_n^\rho,\zeta_n^\rho)
		&=\frac{L_n'}{L_n}
		F(\omega_{n,\mathrm{phys}}^\rho,\zeta_{n,\mathrm{phys}}^\rho).
		\label{eq:fidelity-embedding-scaling}
	\end{align}
	The declared gain is at least $q$ and converges to $q$, while \eqref{eq:diagonal-gain-control} gives
	\begin{align}
		\limsup_{n\to\infty}\frac{1}{n}\log\frac{L_n}{L_n'}
		\leq q-H(A|R)_\tau.
		\label{eq:diagonal-embedding-penalty}
	\end{align}
	Consequently,
	\begin{align}
		\limsup_{n\to\infty}-\frac{1}{n}\log F(\omega_n^\rho,\zeta_n^\rho)
		\leq D(\tau_{AR}\|\rho_{AR})+D(\tau_R\|\rho_R)
		+q-H(A|R)_\tau.
		\label{eq:LE-high-gain}
	\end{align}
	In each regime, the corresponding right-hand side in \eqref{eq:LE-low-gain} or \eqref{eq:LE-high-gain} equals the objective in \eqref{eq:LE-exponent-variational} evaluated at the chosen minimizer.  This proves \eqref{eq:LE-fixed-output-achievability}, while the dimension choices give \eqref{eq:LE-achievable-gain} and \eqref{eq:LE-achievable-gain-limit}.
\end{proof}

\subsection{Controlled Uhlmann recovery}

\begin{lemma}[Controlled Uhlmann recovery]
\label{lem:controlled-Uhlmann}
Let $\mathcal{A}:A^nK_A\to JL_A$ be a one-Kraus-per-outcome instrument applied to
$\psi_{ABR}^{\otimes n}\otimes\Phi_{K_AK_B}$, and define
\begin{align}
 \omega_{JL_AR^n}
 &:=(\mathcal{A}\otimes\mathrm{id}_{R^n})
 (\psi_{AR}^{\otimes n}\otimes\pi_{K_A}),
 \label{eq:controlled-Uhlmann-output}\\
 \zeta_{JL_AR^n}
 &:=\pi_J\otimes\pi_{L_A}\otimes\psi_R^{\otimes n}.
 \label{eq:controlled-Uhlmann-target}
\end{align}
There exist outcome-dependent recovery isometries for Bob that produce a state-merging protocol whose initial and final Schmidt ranks are $|K_A|$ and $|L_A|$ and whose squared fidelity is at least
\begin{align}
 F(\omega,\zeta).
 \label{eq:controlled-Uhlmann-conclusion}
\end{align}
The conclusion remains valid when the physical output of $\mathcal{A}$ is isometrically embedded into a larger declared output space $L_A$.
\end{lemma}

\begin{proof}
	Let $N=|J|$ and write the block-diagonal states as
	\begin{align}
		\omega=\sum_{j=1}^N\ket{j}\!\bra{j}\otimes\omega_{L_AR^n}^j,
		\qquad
		\zeta=\sum_{j=1}^N\ket{j}\!\bra{j}\otimes
		\frac{1}{N}\pi_{L_A}\otimes\psi_R^{\otimes n}.
		\label{eq:controlled-Uhlmann-blocks}
	\end{align}
	Set
	\begin{align}
		f_j:=\norm{\sqrt{\omega^j}
			\sqrt{N^{-1}\pi_{L_A}\otimes\psi_R^{\otimes n}}}_1.
		\label{eq:branch-fidelity}
	\end{align}
	Block diagonality gives
	\begin{align}
		\sqrt{F(\omega,\zeta)}=\sum_{j=1}^Nf_j.
		\label{eq:block-fidelity-sum}
	\end{align}
	Because outcome $j$ is represented by a single Kraus operator, the corresponding subnormalized global post-measurement vector purifies $\omega^j$, and its complementary system is held by Bob.  Uhlmann's theorem~\cite{uhlmann1976transition} therefore provides an isometry on Bob's system whose overlap with
	$N^{-1/2}\Phi_{L_AL_B}\otimes\psi_{\widehat A BR}^{\otimes n}$ is $f_j$.  The branch target has squared norm $1/N$, so the overlap with the corresponding normalized ideal branch is $\sqrt{N}f_j$; summing the squared branch overlaps produces the factor $N$ below.  After the outcome register and decoder environments are discarded, the fidelity with the pure ideal target is at least
	\begin{align}
		N\sum_{j=1}^Nf_j^2
		\geq\left(\sum_{j=1}^Nf_j\right)^2
		=F(\omega,\zeta),
		\label{eq:controlled-Uhlmann-proof}
	\end{align}
	where the inequality follows from the Cauchy--Schwarz inequality.
\end{proof}

\subsection{Double pinching}

The construction in this subsection adapts the double-pinching method of Rubboli and Tomamichel~\cite{rubboli2026strong}, itself embedded in the Log-Euclidean strong-converse strategy of Mosonyi and Ogawa~\cite{mosonyi2017strong}.  The new point required here is to preserve the quantum marginal relevant to state merging while simultaneously commuting the joint state, its marginal, and the auxiliary universal state.

Restrict $R$ to $\supp(\rho_R)$ and write $d_R:=|R|$.  For each block size $m$, let $\omega_{R^m}$ be the universal state in Proposition~\ref{prop:universal-state}, and define
\begin{align}
 \widetilde\rho_{A^mR^m}^{(m)}
 :=\left(\mathrm{id}_{A^m}\otimes
 \P_{\rho_R^{\otimes m}}\circ\P_{\omega_{R^m}}\right)
 (\rho_{AR}^{\otimes m}).
 \label{eq:double-pinched-state}
\end{align}
The state $\widetilde\rho^{(m)}$ is invariant under simultaneous permutations of the $m$ copies.  Since $\omega_{R^m}$ commutes with $\rho_R^{\otimes m}$, their spectral projections commute, and therefore the pinching maps $\P_{\omega_{R^m}}$ and $\P_{\rho_R^{\otimes m}}$ commute.  Consequently,
\begin{align}
 \widetilde\rho_{R^m}^{(m)}=\rho_R^{\otimes m},
 \label{eq:double-pinched-marginal}
\end{align}
and the three operators
\begin{align}
 \widetilde\rho_{A^mR^m}^{(m)},
 \qquad I_{A^m}\otimes\rho_R^{\otimes m},
 \qquad I_{A^m}\otimes\omega_{R^m}
 \label{eq:three-commuting-operators}
\end{align}
commute pairwise.

\begin{lemma}[Log-Euclidean-to-club-sandwiched comparison]
\label{lem:double-pinching-entropy}
Let $\alpha\in[1/2,1)$ and
$\frac{1-2\alpha}{1-\alpha}\leq\lambda\leq0$.  Then
\begin{align}
 H_{\alpha,\infty}^{\lambda}(A^m|R^m)_{\widetilde\rho^{(m)}}
 \geq
 m\widetilde H_\alpha^\lambda(A|R)_\rho
 +\lambda(d_R^2-1)\log(m+1).
 \label{eq:double-pinching-entropy}
\end{align}
\end{lemma}

\begin{proof}
By Proposition~\ref{prop:universal-state}, an optimizer $\sigma_{R^m}^\star$ in the Log-Euclidean conditional entropy can be chosen permutation-invariant and hence satisfies
\begin{align}
 \sigma_{R^m}^\star\leq(m+1)^{d_R^2-1}\omega_{R^m}.
 \label{eq:optimizer-domination}
\end{align}
Operator monotonicity of the logarithm and $\lambda\leq0$ imply
\begin{align}
 \lambda\log\sigma_{R^m}^\star
 \geq\lambda\log\omega_{R^m}
 +\lambda(d_R^2-1)\log(m+1)I_{R^m}.
 \label{eq:negative-lambda-comparison}
\end{align}
Trace exponential is monotone in the L\"owner order: $X\leq Y$ implies $\Tr\exp X\leq\Tr\exp Y$~\cite{bhatia1997matrix}.  This trace monotonicity, rather than operator monotonicity of the matrix exponential, allows us to replace the optimizer by $\omega_{R^m}$ in \eqref{eq:LE-conditional}, with the additive correction shown in \eqref{eq:double-pinching-entropy}.  Because the three operators in \eqref{eq:three-commuting-operators} commute, simultaneous diagonalization gives the common Log-Euclidean/club-sandwiched trace
\begin{align}
 \Tr\!\left[
 \left(
 (\widetilde\rho^{(m)})^{1/2}
 \left(I_{A^m}\otimes
 (\rho_R^{\otimes m})^{\frac{(1-\lambda)(1-\alpha)}{2\alpha}}
 \omega_{R^m}^{\frac{\lambda(1-\alpha)}{\alpha}}
 (\rho_R^{\otimes m})^{\frac{(1-\lambda)(1-\alpha)}{2\alpha}}
 \right)
 (\widetilde\rho^{(m)})^{1/2}
 \right)^\alpha
 \right].
 \label{eq:commuting-LE-club-trace}
\end{align}
Thus the Log-Euclidean expression evaluated at this candidate equals the corresponding club-sandwiched expression.  Minimizing the latter over all states on $R^m$ gives
\begin{align}
 H_{\alpha,\infty}^{\lambda}(A^m|R^m)_{\widetilde\rho^{(m)}}
 \geq
 \widetilde H_\alpha^\lambda(A^m|R^m)_{\widetilde\rho^{(m)}}
 +\lambda(d_R^2-1)\log(m+1).
 \label{eq:LE-to-club-pinched}
\end{align}
Both pinchings act on the conditioning system.  Data processing and additivity of the club-sandwiched conditional entropy now yield
\begin{align}
 \widetilde H_\alpha^\lambda(A^m|R^m)_{\widetilde\rho^{(m)}}
 \geq
 \widetilde H_\alpha^\lambda(A^m|R^m)_{\rho^{\otimes m}}
 =m\widetilde H_\alpha^\lambda(A|R)_\rho,
 \label{eq:club-DPI-additivity}
\end{align}
which proves the claim.  The case of singular optimizers follows by the standard full-rank regularization used in Ref.~\cite{rubboli2026strong}.
\end{proof}

\begin{lemma}[Operational insertion of the double pinching]
\label{lem:operational-pinching}
Let $\mathcal{A}_{m,k}$ be any Alice-side instrument on $A^{mk}K_A$, and let
\begin{align}
 \zeta_{m,k}:=\pi_J\otimes\pi_L\otimes\rho_R^{\otimes mk}.
 \label{eq:block-fixed-target}
\end{align}
Then
\begin{align}
 &-\log F\!\left(
 (\mathcal{A}_{m,k}\otimes\mathrm{id}_{R^{mk}})
 (\rho_{AR}^{\otimes mk}\otimes\pi_{K_A}),
 \zeta_{m,k}\right)
 \nonumber\\
 &\quad\leq
 -\log F\!\left(
 (\mathcal{A}_{m,k}\otimes\mathrm{id}_{R^{mk}})
 ((\widetilde\rho^{(m)})^{\otimes k}\otimes\pi_{K_A}),
 \zeta_{m,k}\right)
 +2k(d_R-1)\log(m+1).
 \label{eq:operational-pinching}
\end{align}
\end{lemma}

\begin{proof}
The instrument acts trivially on $R^{mk}$ and therefore commutes with both blockwise pinchings.  The target in \eqref{eq:block-fixed-target} is invariant under both pinchings by \eqref{eq:double-pinched-marginal} and the commutation of $\rho_R^{\otimes m}$ with $\omega_{R^m}$.  Apply Proposition~\ref{prop:fidelity-pinching} first to $\P_{\omega_{R^m}}^{\otimes k}$ and then to $\P_{\rho_R^{\otimes m}}^{\otimes k}$.  Equation \eqref{eq:spectral-cardinality} bounds the loss in each application by $k(d_R-1)\log(m+1)$.
\end{proof}

\subsection{The achievability proof}

\begin{theorem}[Upper bound on the strong converse exponent]
\label{thm:achievability}
For every pure $\psi_{ABR}$ and every $r\in\mathbb{R}$,
\begin{align}
 \limsup_{n\to\infty}-\frac{1}{n}\log F_n^{\mathrm{merg}}(\psi,r)
 \leq
 \max_{\alpha\in[1,2]}\frac{2(\alpha-1)}{\alpha}
 \left[H_{\alpha,\alpha/2}^{\uparrow}(A|B)_\psi-r\right].
 \label{eq:achievability-exponent}
\end{align}
\end{theorem}

\begin{proof}
	Set $\rho_{AR}=\psi_{AR}$ and $q=-r$.  Fix $m$ and apply Proposition~\ref{prop:LE-fixed-output} to $k$ copies of the block state $\widetilde\rho^{(m)}$ at coherent gain $mq$ per block.  Run the resulting instrument on the original source $\rho_{AR}^{\otimes mk}$, apply Lemma~\ref{lem:operational-pinching}, and then apply Lemma~\ref{lem:controlled-Uhlmann} to a purification of the original source.  This gives a state-merging protocol of cost at most $mkr$ and
	\begin{align}
		&\limsup_{k\to\infty}-\frac{1}{mk}\log F_{mk}^{\mathrm{merg}}(\psi,r)
		\nonumber\\
		&\quad\leq
		\max_{\alpha\in[1/2,1]}
		\frac{1-\alpha}{\alpha}
		\left[q-\frac{1}{m}
		H_{\alpha,\infty}^{\frac{1-2\alpha}{1-\alpha}}
		(A^m|R^m)_{\widetilde\rho^{(m)}}\right]
		+\frac{2(d_R-1)\log(m+1)}{m}.
		\label{eq:block-LE-achievability}
	\end{align}
	
	Apply Lemma~\ref{lem:double-pinching-entropy} along the boundary path, with the endpoint $\alpha=1$ understood by one-sided continuity.  Since
	\begin{align}
		\frac{1-\alpha}{\alpha}
		\left(-\frac{1-2\alpha}{1-\alpha}\right)
		=\frac{2\alpha-1}{\alpha}\leq1,
		\label{eq:boundary-loss}
	\end{align}
	we obtain
	\begin{align}
		&\limsup_{k\to\infty}-\frac{1}{mk}\log F_{mk}^{\mathrm{merg}}(\psi,r)
		\nonumber\\
		&\quad\leq
		\max_{\alpha\in[1/2,1]}
		\frac{1-\alpha}{\alpha}
		\left[q-
		\widetilde H_\alpha^{\frac{1-2\alpha}{1-\alpha}}(A|R)_\rho\right]
		+\frac{(d_R^2+2d_R-3)\log(m+1)}{m}.
		\label{eq:multiples-of-m-bound}
	\end{align}
	
	It remains to define protocols for every blocklength while enforcing the rate without slack.  Write $n=mk+t$ with $0\leq t<m$.  Use the block protocol on the first $mk$ copies and teleport the remaining $t$ copies exactly, consuming $t\log|A|$ ebits.  If the block protocol has initial Schmidt rank $K_{mk}$ and final Schmidt rank $L_{mk}$, the resulting preliminary protocol has ranks
	\begin{align}
		K'_n&:=K_{mk}|A|^t,
		&L'_n&:=L_{mk},
		\label{eq:preliminary-Schmidt-ranks}
	\end{align}
	and hence actual preliminary gain
	\begin{align}
		g'_n:=\log L'_n-\log K'_n\geq mkq-t\log|A|.
		\label{eq:preliminary-gain}
	\end{align}
	Define the smallest enlargement required by the actual preliminary gain,
	\begin{align}
		M_n:=\left\lceil
		2^{[nq-g'_n]_+}
		\right\rceil.
		\label{eq:rate-repair-factor}
	\end{align}
	Take the declared final entanglement registers on both Alice's and Bob's sides to have dimension $L_n:=M_nL'_n$.  Let $V_A$ and $V_B$ be the canonical isometries that embed the $L'_n$-dimensional preliminary final registers into the first $L'_n$ basis vectors of the corresponding $L_n$-dimensional registers.  Applying $V_A\otimes V_B$ to the preliminary output gives
	\begin{align}
		\log L_n-\log K'_n
		\geq nq,
		\label{eq:repaired-gain}
	\end{align}
	so the exact rate constraint is satisfied.  Moreover,
	\begin{align}
		\left\langle\Phi_{L_n}\middle|
		(V_A\otimes V_B)\middle|\Phi_{L'_n}\right\rangle
		=M_n^{-1/2}.
		\label{eq:entanglement-embedding-overlap}
	\end{align}
	Because the embedded output is supported on the ranges of $V_A\otimes V_B$, its squared fidelity with the enlarged target is therefore exactly $M_n^{-1}$ times its squared fidelity with the preliminary target.  Teleportation of the remainder is exact, so it introduces no additional fidelity loss.  Finally, \eqref{eq:preliminary-gain} implies
	\begin{align}
		[nq-g'_n]_+
		&\leq[t(q+\log|A|)]_+
		\leq m\big(|q|+\log|A|\big),
		\nonumber\\
		\log M_n
		&\leq m\big(|q|+\log|A|\big)+1.
		\label{eq:rate-repair-loss}
	\end{align}
	For fixed $m$, the remainder $t<m$ and the logarithmic penalty in \eqref{eq:rate-repair-loss} are bounded independently of $k$.  Thus $mk/n\to1$ and $n^{-1}\log M_n\to0$ as $k\to\infty$, and neither the asymptotic rate nor the fidelity exponent is changed.  Letting $m\to\infty$ in \eqref{eq:multiples-of-m-bound}, substituting $q=-r$, and using Proposition~\ref{prop:club-duality} together with \eqref{eq:parameter-conversion} yields \eqref{eq:achievability-exponent}.
\end{proof}

\begin{proof}[Proof of Theorem~\ref{thm:main}]
	Theorem~\ref{thm:converse} lower-bounds the liminf and
	Theorem~\ref{thm:achievability} upper-bounds the limsup by the same
	quantity.  Therefore the limit exists and
	\eqref{eq:main-conditional-entropy} follows.
	
	If $r<H(A|B)_\psi=-H(A|R)_\psi$, then $q=-r>H(A|R)_\psi$.
	By the one-sided limit in \eqref{eq:club-limit}, the objective in
	\eqref{eq:main-club-form} is positive for some $\alpha<1$ sufficiently
	close to one.  Conversely, if $r\geq H(A|B)_\psi$, then
	\eqref{eq:club-lower-bound} and purity imply
	\begin{align}
		-r-
		\widetilde H_\alpha^{\frac{1-2\alpha}{1-\alpha}}(A|R)_\psi
		\leq-r-H(A|R)_\psi\leq0
		\label{eq:positivity-converse}
	\end{align}
	for every $\alpha<1$, while the endpoint $\alpha=1$ contributes zero.
	This proves \eqref{eq:positive-region}.
\end{proof}

\section{Conclusion and discussion}
\label{sec:conclusion}

We have determined the exact strong converse exponent of quantum state
merging for every finite-dimensional pure state and every net
entanglement-cost rate.  The exponent is positive precisely when the
rate lies below the conditional entropy $H(A|B)_\psi$, thereby
quantifying the exponential decay of the optimal squared fidelity
throughout the strong converse regime.
On the $AB$ system, the exponent is governed by the optimized $\alpha$-$z$ conditional
R\'enyi entropy along the path $z=\alpha/2$, with
$\alpha\in[1,2]$.  This identifies a genuinely two-parameter $\alpha$-$z$ quantity
as the large-deviation rate governing a fundamental quantum
information-processing task.
Equivalently, on the purifying $AR$ system, the exponent is governed by the club-sandwiched
conditional entropy of Ref.~\cite{rubboli2024quantum}.

We also identified the exact strong converse exponent of the partially smoothed conditional min-entropy for every finite-dimensional bipartite state.  
Since partial smoothing imposes the additional constraint $\tau_R\leq\rho_R$, one has $H_{\min}^{\varepsilon,P}(A|\dot R)_{\rho} \leq H_{\min}^{\varepsilon,P}(A|R)_{\rho}$.
This makes a
given conditional min-entropy threshold for $H_{\min}^{\varepsilon,P}(A|\dot R)_{\rho}$ harder to attain.  
The same
ordering appears directly in the single-letter expressions for the
strong converse exponents.  Indeed, choosing $\sigma_R=\rho_R$ in the
variational definition of the club-sandwiched conditional
entropy gives
\begin{align}
	\widetilde H_\alpha^{\frac{1-2\alpha}{1-\alpha}}
	(A|R)_\rho
	\leq
	\widetilde H_\alpha^{\downarrow}(A|R)_\rho,
	\qquad
	\alpha\in[1/2,1].
	\label{eq:club-versus-arrow-down}
\end{align}
Consequently,
\begin{align}
	&\max_{\alpha\in[1/2,1]}
	\frac{1-\alpha}{\alpha}
	\left[
	q-
	\widetilde H_\alpha^{\frac{1-2\alpha}{1-\alpha}}
	(A|R)_\rho
	\right]
	\geq
	\max_{\alpha\in[1/2,1]}
	\frac{1-\alpha}{\alpha}
	\left[
	q-\widetilde H_\alpha^{\downarrow}(A|R)_\rho
	\right].
	\label{eq:partial-global-exponent-order}
\end{align}
Accordingly, the partially smoothed success quantity decays at least as fast as its globally smoothed counterpart, although both quantities have the same first-order threshold $H(A|R)_\rho$.

The difference between the strong converse exponents can be seen at the Log-Euclidean level.  Setting $\lambda=0$ and $\lambda=(1-2\alpha)/(1-\alpha)$, respectively, in Proposition~\ref{prop:LE-conditional-variational}, and applying the same minimax argument as in Lemma~\ref{lem:LE-exponent-variational}, gives the global and partial rate functions
\begin{align}
	\text{global:}\quad
	&\min_{\eta_{AR}}
	\left\{D(\eta_{AR}\|\rho_{AR})+
	\left[q-H(A|R)_\eta+D(\eta_R\|\rho_R)\right]_+\right\},
	\nonumber\\
	\text{partial:}\quad
	&\min_{\eta_{AR}}
	\left\{D(\eta_{AR}\|\rho_{AR})+D(\eta_R\|\rho_R)
	+\left[q-H(A|R)_\eta\right]_+\right\}.
	\label{eq:global-partial-rate-functions}
\end{align}
At the entropy-threshold level, global smoothing imposes only the joint max-relative-entropy cap, so a favorable increase of $H(A|R)_\eta$ can compensate for a mismatch between $\eta_R$ and $\rho_R$.  Partial smoothing forbids this tradeoff: the marginal constraint must be enforced separately, and the full cost $D(\eta_R\|\rho_R)$ is paid independently.  In the general Log-Euclidean variational formula, the coefficient of this marginal divergence is $(1-\lambda)(1-\alpha)/\alpha$.  Requiring it to equal one uniquely selects $\lambda=(1-2\alpha)/(1-\alpha)$, namely the boundary club-sandwiched entropy; global smoothing corresponds instead to $\lambda=0$, which yields the ordinary sandwiched entropy after pinching.  The double-pinching argument transfers this distinction to the fully quantum exponents.  Consequently, whenever the inequality above is strict, the globally smoothed sandwiched exponent is not achievable under partial smoothing.

Quantum state merging is governed by partial smoothing because the reference system $R$ is inaccessible.  Every trace-nonincreasing branch $\mathcal{N}$ acting on Alice's and Bob's systems produces a subnormalized state $\tau$ satisfying $\tau_R\leq\psi_R$, since $\mathcal{N}^*(I)\leq I$.  Hence every physical successful component obeys the fixed-marginal constraint entering $H_{\min}^{\varepsilon,P}(A|\dot R)_\psi$ and the one-shot merging converse of Ref.~\cite{anshu2020partially}.  Global smoothing drops this constraint and is therefore only a relaxation of the operational problem.  The same physical marginal appears in achievability, where decoupling and Uhlmann recovery must reconstruct the original purification with marginal $\psi_R$.

\begin{acknowledgments}

The authors used generative AI, ChatGPT 5.6 Sol, to help organize and assemble the authors’ proof ideas, as well as to prepare the manuscript. This work was supported in part by a grant of access to OpenAI models through the ChatGPT for Academic Researchers program. Figure~\ref{fig:QSM} is generated with the help of Google Gemini.
MB acknowledges support from the European Research Council (ERC Grant Agreement No.~948139) and the Excellence Cluster Matter and Light for Quantum Computing (ML4Q-2).
HC acknowledges support from National Science and Technology Council (NSTC 115-2628-E-002-005, NSTC 114-2119-M-001-002, and NSTC 115-2124-M-002-014) and Ministry of Education (NTU-115V2016-1, NTU-CC115L893705, and NTU-115L900702).
RR acknowledges financial support from the ERC grant
GIFNEQ 101163938.
MT acknowledges support from the National Research Foundation Investigatorship Award
(NRF-NRFI10-2024-0006) and the National Research Foundation, Singapore, through
the National Quantum Office, hosted by A*STAR, under its Centre for Quantum Technologies Funding
Initiative (S24Q2d0009).
\end{acknowledgments}

\appendix
\renewcommand{\thetheorem}{\Alph{section}.\arabic{theorem}}

\section{Exact strong converse exponent of the partially smoothed conditional min-entropy}
\label{app:partial-smoothing-exponent}

Exact partially smoothed strong-converse exponents were previously obtained for pure quantum states and for classical states by Berta and Yao~\cite{berta2025strong}.  The theorem below extends the exponent statement to an arbitrary finite-dimensional, possibly mixed and fully quantum, bipartite state.  In the pure-state limit, it is also consistent with the earlier large-deviation theory surrounding entanglement concentration~\cite{hayashi2003entanglement}; that literature is contextual rather than a substitute for the arbitrary-state theorem proved here.

For $\rho_{AR}\in\mathcal{S}(AR)$ and $q\in\mathbb{R}$, set
\begin{align}
 \Gamma_n(\rho,q)
 :=\Gamma_{A^n|\dot R^n}(\rho_{AR}^{\otimes n},nq)
 \label{eq:Gamma-n}
\end{align}
and define the smallest partial-smoothing error needed to reach the threshold $nq$ by
\begin{align}
 \varepsilon_n^\star(\rho,q)
 :=\inf\left\{\varepsilon\in[0,1]:
 H_{\min}^{\varepsilon,P}(A^n|\dot R^n)_{\rho^{\otimes n}}\geq nq\right\}.
 \label{eq:optimal-partial-smoothing-error}
\end{align}
Equation~\eqref{eq:Hmin-Gamma-equivalence} and compactness give
\begin{align}
 \Gamma_n(\rho,q)=1-\bigl(\varepsilon_n^\star(\rho,q)\bigr)^2.
 \label{eq:Gamma-smoothing-error}
\end{align}

\begin{theorem}[Exact strong converse exponent of the partially smoothed conditional min-entropy]
\label{thm:partial-smoothing-exponent}
For every finite-dimensional bipartite state $\rho_{AR}\in\mathcal{S}(AR)$ and every $q\in\mathbb{R}$,
\begin{align}
 &\lim_{n\to\infty}-\frac{1}{n}\log
 \left[1-\bigl(\varepsilon_n^\star(\rho,q)\bigr)^2\right]
=
 \max_{\alpha\in[1/2,1]}
 \frac{1-\alpha}{\alpha}
 \left[q-
 \widetilde H_\alpha^{\frac{1-2\alpha}{1-\alpha}}(A|R)_\rho\right],
 \label{eq:partial-smoothing-exponent}
\end{align}
where the contribution at $\alpha=1$ is zero by continuity.  In particular, the exponent is positive if and only if $q>H(A|R)_\rho$.
\end{theorem}

The converse part follows from Lemma~\ref{lem:partial-smoothing-one-shot}.  It remains to prove the matching achievability bound.

\begin{proposition}[Log-Euclidean achievability for partial smoothing]
\label{prop:LE-partial-smoothing-achievability}
For every $\rho_{AR}\in\mathcal{S}(AR)$ and $q\in\mathbb{R}$,
\begin{align}
 \limsup_{n\to\infty}-\frac{1}{n}\log\Gamma_n(\rho,q)
 \leq
 \max_{\alpha\in[1/2,1]}
 \frac{1-\alpha}{\alpha}
 \left[q-H_{\alpha,\infty}^{\frac{1-2\alpha}{1-\alpha}}(A|R)_\rho\right].
 \label{eq:LE-partial-smoothing-achievability}
\end{align}
\end{proposition}

\begin{proof}
This proof adapts the Log-Euclidean change-of-measure and pinching architecture of Refs.~\cite{mosonyi2017strong,berta2025strong,rubboli2026strong}; the joint and marginal soft caps below enforce the additional fixed-marginal constraint.  Restrict $R$ to $\supp(\rho_R)$.  Fix $\delta>0$ and $\eta_{AR}\in\mathcal{S}(AR)$ with $\supp(\eta_{AR})\subseteq\supp(\rho_{AR})$.  For each $n$, define
\begin{align}
 \bar\eta_{R^n}&:=\P_{\rho_R^{\otimes n}}(\eta_R^{\otimes n}),
 &
 \mathcal{Q}_n&:=\P_{\bar\eta_{R^n}}\circ\P_{\rho_R^{\otimes n}},
 \label{eq:soft-cap-pinching}\\
 \eta_n&:=(\mathrm{id}_{A^n}\otimes\mathcal{Q}_n)(\eta_{AR}^{\otimes n}).
 \label{eq:soft-cap-eta}
\end{align}
By construction, $\bar\eta_{R^n}$ commutes with $\rho_R^{\otimes n}$.  Their spectral projections therefore commute, and so do the two pinching maps.  Consequently, $(\eta_n)_{R^n}=\bar\eta_{R^n}$, and $\eta_n$ commutes with both $I_{A^n}\otimes\rho_R^{\otimes n}$ and $I_{A^n}\otimes\bar\eta_{R^n}$.  Let $v_n$ be the number of nonzero joint spectral projections in $\mathcal{Q}_n$.  It is at most the product of the two spectral cardinalities.  The state $\rho_R^{\otimes n}$ has at most $(n+1)^{|R|-1}$ distinct eigenvalues, whereas the permutation-invariant operator $\bar\eta_{R^n}$ has at most $\binom{n+|R|^2-1}{|R|^2-1}$ distinct eigenvalues.  Hence
\begin{align}
 \log v_n=o(n).
 \label{eq:soft-cap-polynomial-pinching}
\end{align}

On the common support, define
\begin{align}
 C_n
 &:=2^{-\left[\log\eta_n-\log(I_{A^n}\otimes\rho_R^{\otimes n})
 -n\left(D(\eta_{AR}\|I_A\otimes\rho_R)+\delta\right)I\right]_+},
 \label{eq:joint-soft-cap}\\
 B_n
 &:=2^{-\left[\log \bar\eta_{R^n}-\log\rho_R^{\otimes n}
 -n\left(D(\eta_R\|\rho_R)+\delta\right)I\right]_+}.
 \label{eq:marginal-soft-cap}
\end{align}
All factors commute, and $0<C_n,B_n\leq I$.  The scalar inequality $2^x2^{-[x-a]_+}\leq2^a$, applied by functional calculus, gives
\begin{align}
 \eta_n C_n(I_{A^n}\otimes B_n)
 &\leq
 2^{n(D(\eta_{AR}\|I_A\otimes\rho_R)+\delta)}
 I_{A^n}\otimes\rho_R^{\otimes n},
 \label{eq:joint-cap-bound}\\
 \Tr_{A^n}\!\left[\eta_n C_n(I_{A^n}\otimes B_n)\right]
 &\leq
 2^{n(D(\eta_R\|\rho_R)+\delta)}\rho_R^{\otimes n}.
 \label{eq:marginal-cap-bound}
\end{align}
Indeed, the first inequality follows directly from \eqref{eq:joint-soft-cap}; for the second, use $\eta_nC_n\leq\eta_n$, take the partial trace, and apply \eqref{eq:marginal-soft-cap} to $\bar\eta_{R^n}$.

Now set
\begin{align}
 \tau_n
 :=2^{-n\left(D(\eta_R\|\rho_R)+[q-H(A|R)_\eta]_++\delta\right)}
 \eta_nC_n(I_{A^n}\otimes B_n).
 \label{eq:soft-cap-feasible-operator}
\end{align}
Equation~\eqref{eq:marginal-cap-bound} gives $(\tau_n)_{R^n}\leq\rho_R^{\otimes n}$.  In particular, $\tau_n$ is subnormalized.  Moreover,
\begin{align}
 D(\eta_{AR}\|I_A\otimes\rho_R)
 =D(\eta_R\|\rho_R)-H(A|R)_\eta,
 \label{eq:conditional-relative-entropy-identity}
\end{align}
and therefore \eqref{eq:joint-cap-bound} gives
\begin{align}
 \tau_n\leq2^{-nq}I_{A^n}\otimes\rho_R^{\otimes n}.
 \label{eq:soft-cap-joint-feasibility}
\end{align}
Thus $\tau_n$ is feasible for $\Gamma_n(\rho,q)$.

The operator $\tau_n$ is invariant under $\mathrm{id}_{A^n}\otimes\mathcal{Q}_n$.  Applying Proposition~\ref{prop:fidelity-pinching}, followed by \eqref{eq:fidelity-LE} and the Gibbs formula \eqref{eq:LE-Gibbs} evaluated at $\eta_n$, gives
\begin{align}
 -\log F(\rho_{AR}^{\otimes n},\tau_n)
 &\leq
 D\!\left(\eta_n\middle\|
 (\mathrm{id}_{A^n}\otimes\mathcal{Q}_n)(\rho_{AR}^{\otimes n})\right)
 +D(\eta_n\|\tau_n)+\log v_n
 \nonumber\\
 &\leq nD(\eta_{AR}\|\rho_{AR})+D(\eta_n\|\tau_n)+\log v_n,
 \label{eq:soft-cap-change-of-measure}
\end{align}
where the second inequality follows from data processing.  Since the attenuation operators are strictly positive on $\supp(\eta_n)$,
\begin{align}
 D(\eta_n\|\tau_n)
 &=n\left(D(\eta_R\|\rho_R)+[q-H(A|R)_\eta]_++\delta\right)
 \nonumber\\
 &\quad+
 \Tr\eta_n\left[\log\eta_n-\log(I_{A^n}\otimes\rho_R^{\otimes n})
 -n\left(D(\eta_{AR}\|I_A\otimes\rho_R)+\delta\right)I\right]_+
 \nonumber\\
 &\quad+
 \Tr \bar\eta_{R^n}\left[\log \bar\eta_{R^n}-\log\rho_R^{\otimes n}
 -n\left(D(\eta_R\|\rho_R)+\delta\right)I\right]_+.
 \label{eq:soft-cap-relative-entropy}
\end{align}
The last two terms converge to zero exponentially.  For the first, choose $s\in(0,1]$ so small that
\begin{align}
 D_{1+s,1}(\eta_{AR}\|I_A\otimes\rho_R)
 <D(\eta_{AR}\|I_A\otimes\rho_R)+\frac{\delta}{2}.
 \label{eq:soft-cap-Renyi-choice}
\end{align}
Using $[x]_+\leq2^{sx}/(s\ln2)$, commutativity, data processing of the Petz R\'enyi divergence $D_{1+s,1}$~\cite{petz1986quasi}, and additivity, this term is at most
\begin{align}
 \frac{1}{s\ln2}
 2^{-sn\left(D(\eta_{AR}\|I_A\otimes\rho_R)+\delta
 -D_{1+s,1}(\eta_{AR}\|I_A\otimes\rho_R)\right)}
 \leq\frac{1}{s\ln2}2^{-sn\delta/2}.
 \label{eq:soft-cap-joint-clipping}
\end{align}
By decreasing $s>0$ if necessary, the same $s$ may be chosen so that $D_{1+s,1}(\eta_R\|\rho_R)<D(\eta_R\|\rho_R)+\delta/2$.  The marginal term is then handled analogously using $D_{1+s,1}(\eta_R\|\rho_R)$.  Combining \eqref{eq:soft-cap-polynomial-pinching}, \eqref{eq:soft-cap-change-of-measure}, and \eqref{eq:soft-cap-relative-entropy}, and then letting $n\to\infty$ and $\delta\downarrow0$, yields
\begin{align}
 \limsup_{n\to\infty}-\frac{1}{n}\log\Gamma_n(\rho,q)
 \leq
 D(\eta_{AR}\|\rho_{AR})+D(\eta_R\|\rho_R)
 +[q-H(A|R)_\eta]_+.
 \label{eq:soft-cap-auxiliary-bound}
\end{align}
Minimizing over $\eta_{AR}$ and invoking Lemma~\ref{lem:LE-exponent-variational} proves \eqref{eq:LE-partial-smoothing-achievability}.
\end{proof}

\begin{proof}[Proof of Theorem~\ref{thm:partial-smoothing-exponent}]
For each $n$, the feasible set in $\Gamma_n(\rho,q)$ is a closed subset of the compact set $\mathcal{S}_{\leq}(A^nR^n)$, so the maximum is attained.  It is also strictly positive.  Indeed, $\supp(\rho_{AR}^{\otimes n})\subseteq\supp(I_{A^n}\otimes\rho_R^{\otimes n})$, so one may choose
\begin{align}
 0<c_n\leq
 \min\left\{1,
 2^{-nq-D_{\max}(\rho_{AR}^{\otimes n}\|I_{A^n}\otimes\rho_R^{\otimes n})}
 \right\}.
 \label{eq:Gamma-positive-scaling}
\end{align}
Then $c_n\rho_{AR}^{\otimes n}$ is feasible and
$F(\rho_{AR}^{\otimes n},c_n\rho_{AR}^{\otimes n})=c_n>0$.

Let $\tau_n$ and $\tau_m$ be maximizers for blocklengths $n$ and $m$.  Their tensor product obeys
\begin{align}
 (\tau_n\otimes\tau_m)_{R^{n+m}}
 &\leq\rho_R^{\otimes(n+m)},
 &
 \tau_n\otimes\tau_m
 &\leq2^{-(n+m)q}I_{A^{n+m}}\otimes\rho_R^{\otimes(n+m)},
 \label{eq:Gamma-tensor-feasibility}
\end{align}
so it is feasible at blocklength $n+m$.  Because the first fidelity argument is normalized, squared fidelity is multiplicative under tensor products.  Therefore
\begin{align}
 \Gamma_{n+m}(\rho,q)
 &\geq F(\rho_{AR}^{\otimes(n+m)},\tau_n\otimes\tau_m)
 =\Gamma_n(\rho,q)\Gamma_m(\rho,q).
 \label{eq:Gamma-supermultiplicativity}
\end{align}
Thus $a_n:=-\log\Gamma_n(\rho,q)$ is finite and subadditive.  Fekete's lemma gives
\begin{align}
 \lim_{n\to\infty}\frac{a_n}{n}
 =\inf_{n\geq1}\frac{a_n}{n}.
 \label{eq:Gamma-Fekete-limit}
\end{align}
In particular, the limit in \eqref{eq:partial-smoothing-exponent} exists.  Lemma~\ref{lem:partial-smoothing-one-shot} and additivity of the club-sandwiched conditional entropy give the lower bound on this limit.

For the matching upper bound, fix a block size $m$, let $\widetilde\rho^{(m)}$ be the double-pinched state in \eqref{eq:double-pinched-state}, and apply Proposition~\ref{prop:LE-partial-smoothing-achievability} to this block source at threshold $mq$.  If $\tau$ is feasible for $k$ copies of the block source, apply the two blockwise pinchings to $\tau$.  The resulting operator remains feasible because both pinchings fix $I_{A^m}\otimes\rho_R^{\otimes m}$ and preserve the marginal order.  Applying data processing of fidelity to the pinched source and then Proposition~\ref{prop:fidelity-pinching} twice to the original source gives
\begin{align}
 \Gamma_{mk}(\rho,q)
 \geq(m+1)^{-2k(|R|-1)}
 \Gamma_k(\widetilde\rho^{(m)},mq).
 \label{eq:Gamma-double-pinching-transfer}
\end{align}
Consequently,
\begin{align}
 \lim_{n\to\infty}-\frac{1}{n}\log\Gamma_n(\rho,q)
 &\leq
 \frac{1}{m}\max_{\alpha\in[1/2,1]}
 \frac{1-\alpha}{\alpha}
 \left[mq-H_{\alpha,\infty}^{\frac{1-2\alpha}{1-\alpha}}
 (A^m|R^m)_{\widetilde\rho^{(m)}}\right]
 \nonumber\\
 &\quad+\frac{2(|R|-1)\log(m+1)}{m}.
 \label{eq:Gamma-block-LE-bound}
\end{align}
Lemma~\ref{lem:double-pinching-entropy} and
\begin{align}
 \frac{1-\alpha}{\alpha}
 \left(-\frac{1-2\alpha}{1-\alpha}\right)
 =\frac{2\alpha-1}{\alpha}\leq1
 \label{eq:partial-smoothing-boundary-loss}
\end{align}
therefore yield
\begin{align}
 \lim_{n\to\infty}-\frac{1}{n}\log\Gamma_n(\rho,q)
 &\leq
 \max_{\alpha\in[1/2,1]}
 \frac{1-\alpha}{\alpha}
 \left[q-\widetilde H_\alpha^{\frac{1-2\alpha}{1-\alpha}}(A|R)_\rho\right]
 \nonumber\\
 &\quad+\frac{(|R|^2+2|R|-3)\log(m+1)}{m}.
 \label{eq:Gamma-finite-block-loss}
\end{align}
Letting $m\to\infty$ proves the upper bound and hence \eqref{eq:partial-smoothing-exponent}.  The positivity criterion follows from \eqref{eq:club-limit} and \eqref{eq:club-lower-bound}.
\end{proof}

\bibliography{my}

@article{anshu2020partially,
  author  = {Anshu, Anurag and Berta, Mario and Jain, Rahul and Tomamichel, Marco},
  title   = {Partially smoothed information measures},
  journal = {IEEE Transactions on Information Theory},
  volume  = {66},
  number  = {8},
  pages   = {5022--5036},
  year    = {2020},
  doi     = {10.1109/TIT.2020.2981573}
}

@article{audenaert13_alphaz,
  author  = {Audenaert, Koenraad M. R. and Datta, Nilanjana},
  title   = {$\alpha$-$z$ R{\'e}nyi relative entropies},
  journal = {Journal of Mathematical Physics},
  volume  = {56},
  number  = {2},
  pages   = {022202},
  year    = {2015},
  doi     = {10.1063/1.4906367}
}

@article{berta2026tight,
  author  = {Berta, Mario and Cheng, Hao-Chung and Yao, Yongsheng},
  title   = {Tight any-shot quantum decoupling},
  journal = {arXiv:2602.17430},
  year    = {2026}
}

@book{bhatia1997matrix,
  author    = {Bhatia, Rajendra},
  title     = {Matrix Analysis},
  series    = {Graduate Texts in Mathematics},
  volume    = {169},
  publisher = {Springer},
  address   = {New York},
  year      = {1997},
  doi       = {10.1007/978-1-4612-0653-8}
}

@article{hayashitomamichel15c,
  author  = {Hayashi, Masahito and Tomamichel, Marco},
  title   = {Correlation detection and an operational interpretation of the R{\'e}nyi mutual information},
  journal = {Journal of Mathematical Physics},
  volume  = {57},
  number  = {10},
  pages   = {102201},
  year    = {2016},
  doi     = {10.1063/1.4964755}
}

@article{horodecki2005partial,
  author  = {Horodecki, Micha{\l} and Oppenheim, Jonathan and Winter, Andreas},
  title   = {Partial quantum information},
  journal = {Nature},
  volume  = {436},
  number  = {7051},
  pages   = {673--676},
  year    = {2005},
  doi     = {10.1038/nature03909}
}

@article{rubboli2024quantum,
  author  = {Rubboli, Roberto and Goodarzi, Milad M. and Tomamichel, Marco},
  title   = {Quantum conditional entropies from convex trace functionals},
  journal = {arXiv:2410.21976},
  year    = {2024}
}

@article{rubboli2026strong,
  author  = {Rubboli, Roberto and Tomamichel, Marco},
  title   = {The strong converse exponent of composable randomness extraction against quantum side information},
  journal = {arXiv:2601.19182},
  year    = {2026}
}

@article{sion1958general,
  author  = {Sion, Maurice},
  title   = {On general minimax theorems},
  journal = {Pacific Journal of Mathematics},
  volume  = {8},
  number  = {1},
  pages   = {171--176},
  year    = {1958},
  doi     = {10.2140/pjm.1958.8.171}
}

@book{tomamichel16_book,
  author    = {Tomamichel, Marco},
  title     = {Quantum Information Processing with Finite Resources},
  series    = {SpringerBriefs in Mathematical Physics},
  volume    = {5},
  publisher = {Springer International Publishing},
  year      = {2016},
  doi       = {10.1007/978-3-319-21891-5}
}

@article{uhlmann1976transition,
  author  = {Uhlmann, Armin},
  title   = {The transition probability in the state space of a $*$-algebra},
  journal = {Reports on Mathematical Physics},
  volume  = {9},
  number  = {2},
  pages   = {273--279},
  year    = {1976},
  doi     = {10.1016/0034-4877(76)90060-4}
}

@article{petz1986quasi,
  author  = {Petz, D{\'e}nes},
  title   = {Quasi-entropies for finite quantum systems},
  journal = {Reports on Mathematical Physics},
  volume  = {23},
  number  = {1},
  pages   = {57--65},
  year    = {1986},
  doi     = {10.1016/0034-4877(86)90067-4}
}

@article{horodecki2007merging,
  author  = {Horodecki, Micha{\l} and Oppenheim, Jonathan and Winter, Andreas},
  title   = {Quantum state merging and negative information},
  journal = {Communications in Mathematical Physics},
  volume  = {269},
  number  = {1},
  pages   = {107--136},
  year    = {2007},
  doi     = {10.1007/s00220-006-0118-x}
}

@article{abeyesinghe2009mother,
  author  = {Abeyesinghe, Anura and Devetak, Igor and Hayden, Patrick and Winter, Andreas},
  title   = {The mother of all protocols: restructuring quantum information's family tree},
  journal = {Proceedings of the Royal Society A: Mathematical, Physical and Engineering Sciences},
  volume  = {465},
  number  = {2108},
  pages   = {2537--2563},
  year    = {2009},
  doi     = {10.1098/rspa.2009.0202}
}

@article{berta2009single,
  author  = {Berta, Mario},
  title   = {Single-shot quantum state merging},
  journal = {arXiv:0912.4495},
  year    = {2009}
}

@article{sharma2014strong,
  author  = {Sharma, Naresh},
  title   = {A strong converse for the quantum state merging protocol},
  journal = {arXiv:1404.5940},
  year    = {2014}
}

@article{mosonyiogawa2015hypothesis,
  author  = {Mosonyi, Mil{\'a}n and Ogawa, Tomohiro},
  title   = {Quantum hypothesis testing and the operational interpretation of the quantum R{\'e}nyi relative entropies},
  journal = {Communications in Mathematical Physics},
  volume  = {334},
  number  = {3},
  pages   = {1617--1648},
  year    = {2015},
  doi     = {10.1007/s00220-014-2248-x}
}

@article{leditzky2016strong,
  author  = {Leditzky, Felix and Wilde, Mark M. and Datta, Nilanjana},
  title   = {Strong converse theorems using R{\'e}nyi entropies},
  journal = {Journal of Mathematical Physics},
  volume  = {57},
  number  = {8},
  pages   = {082202},
  year    = {2016},
  doi     = {10.1063/1.4960099}
}

@article{mosonyi2017strong,
  author  = {Mosonyi, Mil{\'a}n and Ogawa, Tomohiro},
  title   = {Strong converse exponent for classical-quantum channel coding},
  journal = {Communications in Mathematical Physics},
  volume  = {355},
  number  = {1},
  pages   = {373--426},
  year    = {2017},
  doi     = {10.1007/s00220-017-2928-4}
}

@article{li2023strong,
  author  = {Li, Ke and Yao, Yongsheng},
  title   = {Strong converse exponent for entanglement-assisted communication},
  journal = {IEEE Transactions on Information Theory},
  volume  = {70},
  number  = {7},
  pages   = {5017--5029},
  year    = {2024},
  doi     = {10.1109/TIT.2023.3335219}
}

@article{li2024operational,
  author  = {Li, Ke and Yao, Yongsheng},
  title   = {Operational interpretation of the sandwiched R{\'e}nyi divergence of order $1/2$ to $1$ as strong converse exponents},
  journal = {Communications in Mathematical Physics},
  volume  = {405},
  number  = {2},
  pages   = {22},
  year    = {2024},
  doi     = {10.1007/s00220-023-04890-8}
}

@article{li2022tight,
  author  = {Li, Ke and Yao, Yongsheng and Hayashi, Masahito},
  title   = {Tight exponential analysis for smoothing the max-relative entropy and for quantum privacy amplification},
  journal = {IEEE Transactions on Information Theory},
  volume  = {69},
  number  = {3},
  pages   = {1680--1694},
  year    = {2023},
  doi     = {10.1109/TIT.2022.3217671}
}

@article{berta2024strong,
  author  = {Berta, Mario and Yao, Yongsheng},
  title   = {Strong converse exponent of quantum dichotomies},
  journal = {arXiv:2410.12576},
  year    = {2024}
}

@article{berta2025strong,
  author  = {Berta, Mario and Yao, Yongsheng},
  title   = {Strong converse exponents of partially smoothed information measures},
  journal = {arXiv:2505.06050},
  year    = {2025}
}

@article{verhagen2026conditions,
  author  = {Verhagen, Frits and Tomamichel, Marco and Haapasalo, Erkka},
  title   = {Conditions for large-sample majorization of pairs of flat states in terms of $\alpha$-$z$ relative entropies},
  journal = {Communications in Mathematical Physics},
  volume  = {407},
  number  = {7},
  pages   = {157},
  year    = {2026},
  doi     = {10.1007/s00220-026-05675-5}
}

@article{Wilde3,
  title={Strong converse for the classical capacity of entanglement-breaking and Hadamard channels via a sandwiched R{\'e}nyi relative entropy},
  author={Wilde, Mark M and Winter, Andreas and Yang, Dong},
  journal={Communications in Mathematical Physics},
  volume={331},
  number={2},
  pages={593--622},
  year={2014},
  publisher={Springer},
  doi={10.1007/s00220-014-2122-x}
}

@article{MDS+13,
	doi = {10.1063/1.4838856},
	year = 2013,
	publisher = {{AIP} Publishing},
	volume = {54},
	number = {12},
	pages = {122203},
	author = {Martin M{\"u}ller-Lennert and Fr{\'e}d{\'e}ric Dupuis and Oleg Szehr and Serge Fehr and Marco Tomamichel},
	title = {On quantum {R{\'e}nyi} entropies: A new generalization and some properties},
	journal = {Journal of Mathematical Physics}
}

@article{CHDH-2018,
  author  = {Cheng, Hao-Chung and Hanson, Eric P. and Datta, Nilanjana and Hsieh, Min-Hsiu},
  title   = {Non-asymptotic classical data compression with quantum side information},
  journal = {IEEE Transactions on Information Theory},
  volume  = {67},
  number  = {2},
  pages   = {902--930},
  year    = {2021},
  doi     = {10.1109/TIT.2020.3038517}
}

@article{CHDH2-2018,
  author  = {Cheng, Hao-Chung and Hanson, Eric P. and Datta, Nilanjana and Hsieh, Min-Hsiu},
  title   = {Duality between source coding with quantum side information and classical--quantum channel coding},
  journal = {IEEE Transactions on Information Theory},
  volume  = {68},
  number  = {11},
  pages   = {7315--7345},
  year    = {2022},
  doi     = {10.1109/TIT.2022.3182748},
  eprint  = {1809.11143},
  archivePrefix = {arXiv},
  primaryClass = {quant-ph}
}

@article{araki1990inequality,
  author  = {Araki, Huzihiro},
  title   = {On an inequality of Lieb and Thirring},
  journal = {Letters in Mathematical Physics},
  volume  = {19},
  number  = {2},
  pages   = {167--170},
  year    = {1990},
  doi     = {10.1007/BF01045887}
}

@article{petz1988variational,
  author  = {Petz, D{\'e}nes},
  title   = {A variational expression for the relative entropy},
  journal = {Communications in Mathematical Physics},
  volume  = {114},
  number  = {2},
  pages   = {345--349},
  year    = {1988},
  doi     = {10.1007/BF01225040}
}

@article{hayashi2003entanglement,
  author  = {Hayashi, Masahito and Koashi, Masato and Matsumoto, Keiji and Morikoshi, Fumiaki and Winter, Andreas},
  title   = {Error exponents for entanglement concentration},
  journal = {Journal of Physics A: Mathematical and General},
  volume  = {36},
  number  = {2},
  pages   = {527--553},
  year    = {2003},
  doi     = {10.1088/0305-4470/36/2/316},
  eprint  = {quant-ph/0206097},
  archivePrefix = {arXiv}
}

\end{document}